\documentclass[sigplan]{acmart}
\renewcommand\footnotetextcopyrightpermission[1]{} 

\acmConference[PL'18]{ACM SIGPLAN Conference on Programming Languages}{January 01--03, 2018}{New York, NY, USA}
\acmYear{2018}
\acmISBN{} 
\acmDOI{} 
\startPage{1}

\setcopyright{none}

\usepackage{booktabs}   
\usepackage{subcaption} 

\usepackage{array}
\usepackage{multirow}
\usepackage{soul}

\usepackage{algorithm}
\usepackage[noend]{algpseudocode}
\usepackage{float}
\usepackage{amsmath}
\usepackage{graphicx}
\usepackage{amsfonts}
\usepackage{amsthm}
\usepackage{pifont}
\usepackage{xspace}
\usepackage{mathtools}
\usepackage{enumitem}
\usepackage{lineno}
\usepackage{thmtools}
\usepackage{thm-restate}
\usepackage{hhline}
\usepackage{abstract}
\usepackage{appendix}
\usepackage{subcaption}
\usepackage{xspace}
\usepackage{graphicx}
\usepackage{amsmath}
\usepackage{xcolor}
\usepackage{url}

\usepackage{pgfplots}
\pgfplotsset{compat=1.17} 
\usepackage{tikz}
\usetikzlibrary{patterns}

\usepackage{thmtools}
\usepackage{thm-restate}

\usepackage{hyperref}
\hypersetup{
    colorlinks = true,
    linkcolor = {blue},
}
\usepackage{cleveref}

\algnewcommand{\LineComment}[1]{\State \(\triangleright\) #1}

\newcommand\CONDITION[2]%
  {\begin{tabular}[t]{@{}l@{}l@{}}
     #1&#2
   \end{tabular}%
  }
\algdef{SE}[WHILE]{While}{}[1]%
  {\algorithmicwhile\ \CONDITION{#1}{\ \algorithmicdo}}%
\algdef{SE}[FOR]{For}{}[1]%
  {\algorithmicfor\ \CONDITION{#1}{\ \algorithmicdo}}%
\algdef{S}[FOR]{ForAll}[1]%
  {\algorithmicforall\ \CONDITION{#1}{\ \algorithmicdo}}
\algdef{SE}[REPEAT]{Repeat}{Until}{\algorithmicrepeat}[1]%
  {\algorithmicuntil\ \CONDITION{#1}{}}
\algdef{SE}[IF]{If}{}[1]%
  {\algorithmicif\ \CONDITION{#1}{\ \algorithmicthen}}%
\algdef{C}[IF]{IF}{ElsIf}[1]%
  {\algorithmicelse\ \algorithmicif\ \CONDITION{#1}{\ \algorithmicthen}}

\DeclareMathOperator*{\argmax}{arg\,max}

\newcommand{\BSTM}{Block-STM\xspace}

\newcommand{\rust}[1]{\ensuremath{\mathtt{#1}}\xspace}
\newcommand{\module}[1]{\textcolor{orange}{\ensuremath{\mathbf{\mathsf{#1}}}\xspace}}

\setlist[description]{leftmargin=\parindent,labelindent=\parindent}

\newcommand{\alglinenoNew}[1]{\newcounter{ALG@line@#1}}
\newcommand{\alglinenoPop}[1]{\setcounter{ALG@line}{\value{ALG@line@#1}}}
\newcommand{\alglinenoPush}[1]{\setcounter{ALG@line@#1}{\value{ALG@line}}}

\newcommand{\com}[1]{}

\algdef{SE}[Upon]{Upon}{EndUpon}[1]{\textbf{upon}\ #1\ \algorithmicdo}{\algorithmicend\ \textbf{}}%
\algtext*{EndUpon}

\newcommand\StateX{\Statex\hspace{\algorithmicindent}}

\algrenewcommand\textproc{}

\newtheorem{theorem}{Theorem}
\newtheorem{claim}{Claim}
\newtheorem{lemma}{Lemma}

\newtheorem{corollary}{Corollary}
\newtheorem{definition}{Definition}

\newcommand{\namedref}[2]{\hyperref[#2]{#1~\ref*{#2}}}
\newcommand{\theoremref}[1]{\namedref{Theorem}{#1}}
\newcommand{\equationref}[1]{\hyperref[#1]{(\ref*{#1})}}
\newcommand{\lemmaref}[1]{\namedref{Lemma}{#1}}
\newcommand{\claimref}[1]{\namedref{Claim}{#1}}

\newcommand{\corollaryref}[1]{\namedref{Corollary}{#1}} 
 
\newcommand{\figureref}[1]{\namedref{Figure}{#1}}
\newcommand{\algorithmref}[1]{\namedref{Algorithm}{#1}}

\newcommand{\sectionref}[1]{\namedref{Section}{#1}} 
\newcommand{\appref}[1]{\namedref{Appendix}{#1}} 
\newcommand{\definitionref}[1]{\namedref{Definition}{#1}} 
\let\lineref\linerefa
\newcommand{\lineref}[1]{\namedref{Line}{#1}}

\algdef{SE}[SUBALG]{Indent}{EndIndent}{}{\algorithmicend\ }%
\algtext*{Indent}
\algtext*{EndIndent}

\newcommand{\omitit}[1]{}

\usepackage[skip=0.8pt,font=small,labelfont=bf]{caption}
\usepackage{multicol}
\usepackage{enumitem}
\setlist{nolistsep}

\begin{document}

\title{\BSTM}         
\titlenote{Rati Gelashvili, Alexander Spiegelman, and Zhuolun Xiang share first authorship. 
Contact emails: \url{gelash@aptoslabs.com}, \url{sasha.spiegelman@gmail.com}, \url{xiangzhuolun@gmail.com}.
The work was initiated while all authors were working at Novi at Meta.}         
\subtitle{Scaling Blockchain Execution by Turning Ordering Curse to a Performance Blessing}                     


\author{Rati Gelashvili}
\affiliation{
  \institution{Aptos}            
}

\author{Alexander Spiegelman}
\affiliation{
  \institution{Aptos}            
}
\email{}          

\author{Zhuolun Xiang}
\affiliation{
  \institution{Aptos}            
}

\author{George Danezis}
\affiliation{%
  \institution{Mysten Labs \& UCL}
}
\author{Zekun Li}
\affiliation{%
  \institution{Aptos}
}
\author{Dahlia Malkhi}
\affiliation{%
  \institution{Chainlink Labs}
  }
\author{Yu Xia}
\affiliation{%
  \institution{MIT}
}
\author{Runtian Zhou}
\affiliation{%
  \institution{Aptos}
}



\maketitle

\pagestyle{plain} 

\begin{abstract}

\BSTM is a parallel execution engine for smart contracts, built around the principles of Software Transactional Memory. 
Transactions are grouped in blocks, and every execution of the block must yield the same deterministic outcome.
\BSTM further enforces that the outcome is consistent with executing transactions according to a preset order, leveraging this order to dynamically detect dependencies and avoid conflicts during speculative transaction execution.
At the core of \BSTM is a novel, low-overhead collaborative scheduler of execution and validation tasks.

\BSTM is implemented on the main branch of the Diem Blockchain code-base and runs in production at Aptos.
Our evaluation demonstrates that \BSTM is adaptive to workloads with different conflict rates and utilizes the inherent parallelism therein.
\BSTM achieves up to $110k$ tps in the Diem benchmarks and up to $170k$ tps in the Aptos Benchmarks, which is a $20$x and $17$x improvement over the sequential baseline with $32$ threads, respectively.
The throughput on a contended workload is up to $50k$ tps and $80k$ tps in Diem and Aptos benchmarks, respectively.
\end{abstract}

\section{Introduction}

A central challenge facing emerging decentralized web3 platforms and applications is improving the throughput of the underlying Blockchain systems.
At the core of a Blockchain system is state machine replication, allowing
a set of entities to agree on and apply a sequence of \emph{blocks} of transactions.
Each transaction contains \emph{smart contract} code to be executed, and every entity that executes the block of transactions must arrive at the same final state. 
While there has been progress on scaling parts of the system, 
Blockchains are still bottlenecked by other components, such as transaction execution.

Our goal is to accelerate the in-memory execution of transactions via parallelism. 
Transactions that access different memory locations 
can always be executed in parallel.
However, in a Blockchain system transactions can have significant number of access conflicts.
This may happen due to potential performance attacks, accessing popular contracts or due to economic opportunities (such as auctions and arbitrage~\cite{flash}).

Conflicts are the main challenge for performance.
An approach pioneered by Software Transactional Memory (STM) libraries~\cite{herlihy1993transactional,shavit1997software}
is to instrument 
memory accesses to detect conflicts.
STM libraries with optimistic concurrency control~\cite{dice2006transactional} (OCC) record memory accesses, \emph{validate} every transaction post execution, and abort and re-execute transactions when validation surfaces a conflict.
The final outcome is equivalent to executing transactions sequentially in some order.
This equivalent order is called \textit{serialization}.

Prior works~\cite{dickerson2020adding, amiri2019parblockchain, anjana2021optsmart}
have capitalized on the specifics of the Blockchain use-case to improve on the STM performance.
Their approach is to pre-compute dependencies in a form of a directed acyclic graph of transactions that can be executed via a fork-join schedule.
The resulting schedule is dependency-aware, and avoids corresponding conflicts.
If  entities are incentivized to record and share the dependency graph, then some entities may be able to avoid the pre-computation overhead.

In the context of deterministic databases, \textsc{Bohm}~\cite{faleiro2015rethinking} demonstrated a way to avoid pre-computing the dependency graph.
\textsc{Bohm} assumes that the write-sets of all transactions are known prior to execution, and enforces a specific preset serialization of transactions.
As a result, each read is associated with the last write preceding it in that order.
Using a multi-version data-structure~\cite{bernstein1983multiversion}, \textsc{Bohm} executes transactions when their read dependencies are resolved, avoiding corresponding conflicts.

\textbf{Our contribution.}
We present \BSTM, an in-memory smart contract parallel execution engine built around the principles of optimistically controlled STM. 
\BSTM does not require a priori knowledge of transaction write-sets, avoids pre-computation, 
and accelerates transaction execution autonomously without requiring further communication. Similar to \textsc{Bohm}, \BSTM uses multi-version shared data-structure
and enforces a preset serialization. 
The final outcome is equivalent to the sequential execution of transactions in the preset order in which they appear in the block.

The key observation is that with OCC and a preset serialization, when a transaction aborts, its write-set
can be used to efficiently detect future dependencies. This has two advantages with respect to pre-execution:
(1) in the optimistic case when there are few conflicts, most transactions are executed once, (2) otherwise,  write-sets are likely to be more accurate as they are based on a more up-to-date execution.
Anther advantage of the of the preset order is that it allows as comprehensive correctness testing as we can compare to a sequential execution output.

Two observations that contribute to the performance of \BSTM in the Blockchain context are the following.
First, in blockchain systems, the state is updated per block.
This allows the \BSTM to avoid the synchronization cost of committing transactions individually.
Instead \BSTM lazily commits all transactions in a block 
based on two atomic counters and a double-collect technique~\cite{attiya2004distributed}. 
Second, transactions are specified in smart contract languages, such as Move~\cite{blackshear2019move} and Solidity~\cite{wohrer2018smart}, and run in a virtual machine that encapsulates their execution and ensures safe behavior.
Therefore, opacity~\cite{guerraoui2007opacity} is not required, allowing \BSTM to efficiently combine an optimistic concurrent control with multi-version data structure, without additional mechanisms to avoid reaching inconsistent states.

The main challenge in combining OCC and preset serialization is that validations are no longer independent from each other and must logically occur in a sequence.
A failed validation of a transaction implies that all higher transactions can be committed only if they get successfully validated afterwards. 
\BSTM handles this issue via a novel collaborative scheduler 
that optimistically dispatches execution and validation tasks, prioritizing tasks for transactions lower in the preset order.
While concurrent priority queues are notoriously hard to scale across threads~\cite{spraylist, multiqueue}, \BSTM capitalizes on the preset serialization order and the boundedness of transaction indices to implement a concurrent ordered set abstraction using only a few shared atomic counters.




We provide comprehensive correctness proofs for both Safety and Liveness, proving that no deadlock or livelock is possible and the final state is always equivalent to the state produced by executing the transactions sequentially. 

A Rust implementation of \BSTM is merged on the main branches
of the Diem~\cite{diem} and its successor Aptos~\cite{aptoswhitepaper} open source blockchain code-bases~\cite{diemcodebase, aptoscodebase}. 
The experimental evaluation demonstrates that \BSTM outperforms sequential execution by up to $20$x on low-contention workloads and by up to $9$x on high-contention ones.
Importantly, \BSTM suffers from at most 30\% overhead when the workload is completely sequential.
In addition, \BSTM significantly outperforms a state-of-the-art deterministic STM~\cite{xia2019litm} implementation, and performances closely to \textsc{Bohm} which requires perfect write-sets information prior to execution.
 

The rest of the paper is organized as following: \sectionref{sec:oview} provides a high-level overview of \BSTM. 
\sectionref{sec:details} describes the full algorithm, while \sectionref{sec:ev} describes \BSTM implementation and evaluation.
\sectionref{sec:related} discusses related work and \sectionref{sec:disc} concludes the paper.
\appref{sec:proof} contains the comprehensive correctness proofs.

\begin{figure*}[ht]
\footnotesize
\begin{flushleft}
    \textbf{Check done:} if 
    $V$ and $E$ are empty and no other thread is performing a task, then return.

    \textbf{Find next task:} Perform the task with the smallest transaction index $tx$ in $V$ and $E$:
\end{flushleft}    
    \begin{enumerate}
        \item \textbf{Execution task:} Execute the next incarnation of $tx$. If a value marked as \textsc{estimate} is read, abort execution and add $tx$ back to $E$. Otherwise:
        %
        \begin{enumerate}
            \item[(a)] If there is a write to a memory location to which the previous finished incarnation of $tx$ has not written, create validation tasks for all transactions $\geq tx$ that are not currently in $E$ or being executed and add them to $V$.
            \item[(b)]  Otherwise, create a validation task only for $tx$ and add it to $V$.
        \end{enumerate} 
        \item \textbf{Validation task:} Validate the last incarnation of $tx$. If validation succeeds, continue. Otherwise, \textbf{abort:}
         \begin{enumerate}
            \item[(a)] Mark every value (in the multi-versioned data-structure) written by the incarnation (that failed validation) as an \textsc{estimate}.
            \item[(b)] Create validation tasks for all transactions $> tx$ that are not currently in $E$ or being executed and add them to $V$. 
            \item[(c)] Create an execution task for transaction $tx$ with an incremented incarnation number, and add it to $E$.
            %
         %
         \end{enumerate}
    \end{enumerate}
\caption{High level scheduling}
\label{fig:highsched}
\end{figure*}

\section{Overview}
\label{sec:oview}
\begin{figure*}[ht]
\includegraphics[scale=0.45]{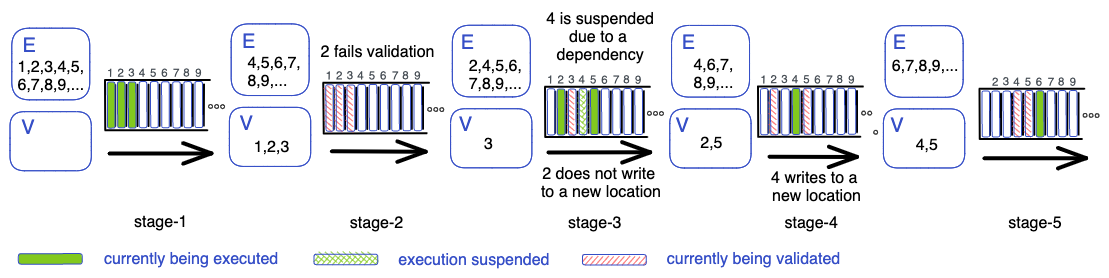}

Illustration of an example execution of the abstract \BSTM collaborative scheduler.
{
\footnotesize 
\begin{flushleft} Initially, all transactions are in the ordered set $E$. In this example, transaction $tx_4$ depends on $tx_2$.
In stage 1, since there are no validation tasks, the threads execute transactions $tx_1, tx_2, tx_3$ in parallel.
Then, in stage 2, the threads validate transactions $tx_1, tx_2, tx_3$ in parallel, the validation of $tx_2$ fails and the validations of $tx_1$ and $tx_3$ succeed.
The incarnation of $tx_2$ is aborted, each of its writes is marked as an \textsc{estimate} in the multi-version data-structure, the next incarnation task is added to $E$, and a new validation task for $tx_3$ is added to $V$.
In stage 3, transaction $tx_3$ is validated and transactions $tx_2$ and $tx_4$ start executing their respective incarnations.  
However, the execution of $tx_4$ reads a value marked as \textsc{estimate}, is aborted due to the dependency on $tx_2$ and the thread executes the next transaction in $E$, which is $tx_5$.
As explained above, $tx_4$ is recorded as a dependency of $tx_2$ 
and added back to $E$ when $tx_2$'s incarnation finishes. 
After both $tx_2$ and $tx_5$ finish execution, the corresponding validation tasks are added to $V$.
In this example, the incarnation of $tx_2$ does not write to a memory location to which its previous incarnation did not write.
Therefore, another validation of $tx_3$ is not required.
In stage 4, $tx_2$ and $tx_5$ are successfully validated and $tx_4$ is executed.
From this point on, $tx_1, tx_2$, and $tx_3$ will never be re-executed as there is no task associated with them in $V$ or $E$ (and no task associated with a higher transaction may lead to creating it).
The execution of $tx_4$ writes to a new memory location, and thus $tx_5$ is added to $V$ for re-validation.
In stage 5, transactions $tx_4$ and $tx_5$ are validated and transaction $tx_6$ is executed.
\end{flushleft}
}
\label{fig:scedular}
\end{figure*}

%
%
The input of \BSTM is a block of transactions, denoted by \rust{BLOCK}, containing $n$ transactions, 
which defines the preset serialization order $tx_1<tx_2<...<tx_n$.
The problem definition is to execute the block and produce the final state 
equivalent to the state produced by executing the transactions in sequence $tx_1$, $tx_2$, \dots $tx_n$, each $tx_j$  executed to completion before $tx_{j+1}$ is started.
The goal is to utilize available threads to produce such final state as efficiently as possible.

Each transaction in \BSTM might be executed several times and we refer to the $i^{th}$ execution as \emph{incarnation} $i$ of a transaction.
We say that an incarnation is \emph{aborted} when the system decides that a subsequent re-execution with an incremented incarnation number is needed.
A \emph{version} is a pair of a transaction index and an \emph{incarnation number}.
To support reads and writes by transactions that may execute concurrently, 
\BSTM 
maintains an in-memory multi-version data structure
that separately stores for each memory location the latest value written per transaction, along with the associated transaction version.
When transaction $tx$ reads a memory location, it obtains from the multi-version data-structure the value written to this location by the highest transaction that appears before $tx$ in the preset serialization order, along with the associated version.
For example, transaction $tx_5$ can read a value written by transaction $tx_3$ even if transaction $tx_6$ has written to same location.
If no smaller transaction has written to a location, then the read (e.g. all reads by $tx_1$) is resolved from storage based on the state before the block execution.

For each incarnation, \BSTM maintains a \emph{write-set} and a \emph{read-set}.
The read-set contains the memory locations that are read during the incarnation, and the corresponding versions. The write-set describes the updates made by the incarnation as (memory location, value) pairs. The write-set of the incarnation is applied to shared memory (the multi-version data-structure) at the end of execution.
After an incarnation executes it needs to pass validation. 
The validation re-reads the read-set and compares the observed versions. 
Intuitively, a successful validation implies that writes applied by the incarnation are still up-to-date, while a failed validation implies 
that the incarnation has to be aborted.

\textbf{Dependency estimation.}
\BSTM does not pre-compute dependencies. Instead, for each transaction, \BSTM treats the write-set of an aborted incarnation as an estimation of the write-set of the next one. Together with the multi-version data structure and the preset order it allows reducing the abort rate by efficiently detecting potential dependencies.
When an incarnation is aborted due to a validation failure, the entries in the multi-version data-structure corresponding to its write-set are replaced with a special \textsc{estimate} marker.
This signifies that the next incarnation is estimated to write to the same memory location.
In particular, an incarnation of transaction $tx_j$ stops and is immediately aborted whenever it reads a value marked as an \textsc{estimate} that was written by a lower transaction $tx_k$.
This is an optimization to abort an incarnation early when it is likely to be aborted in the future due to a validation failure, which would happen if the next incarnation of $tx_k$ would indeed write to the same location (the \rust{ESTIMATE} markers that are not overwritten are removed by the next incarnation).

\textbf{Collaborative scheduler.} 
\BSTM introduces a collaborative scheduler, which coordinates the validation and execution tasks among threads.
The preset serialization order dictates that the transactions must be committed in order, so a successful validation of an incarnation does not guarantee that it can be committed.
This is because an abort and re-execution of an earlier transaction in the block might invalidate the incarnation read-set and necessitate re-execution.
Thus, when a transaction aborts, all higher transactions are scheduled for re-validation.
The same incarnation may be validated multiple times, by different threads, and potentially in parallel, but \BSTM ensures that only the first abort per version succeeds (the rest are ignored).

Since transactions must be committed in order, the \BSTM scheduler prioritizes tasks (validation and execution) associated with lower-indexed transactions. 
Next, we overview the high-level ideas behind the approach. The detailed logic is described in~\sectionref{sec:details} and formally proved in \appref{sec:proof}.

Abstractly, the \BSTM collaborative scheduler tracks an ordered set $V$ of pending validation tasks and an ordered set $E$ of pending execution tasks.
Initially, $V$ is empty and $E$ contains execution tasks for the initial incarnation of all transactions in the block.
A transaction $tx \not \in E$ is either currently being executed or (its last incarnation) has completed.
On a high level, each thread repeats the instructions described in~\figureref{fig:highsched}.

%
When a transaction $tx_k$ reads an \rust{ESTIMATE} marker written by $tx_j$ (with $j < k$), we say that $tx_k$ encounters a \emph{dependency}. 
We treat $tx_k$ as $tx_j$'s dependency because its read depends on a value that $tx_j$ is estimated to write.
For the ease of presentation, in the above description a transaction is added back to $E$ immediately upon encountering a dependency.
However, as explained in~\sectionref{sec:details}, \BSTM implements a slightly more involved mechanism. 
Transaction $tx_k$ is first recorded separately as a dependency of $tx_j$,
and only added back to $E$ when the next incarnation of $tx_j$ completes (i.e. when the dependency is resolved). 

The ordered sets, $V$ and $E$, are each implemented via a single atomic counter coupled with a mechanism to track the status of transactions, i.e. whether a given transaction is ready for validation or execution, respectively.
To pick a task, threads increment the smaller of these counters until they find a task that is ready to be performed.
To add a (validation or execution) task for transaction $tx$, the thread updates the status and reduces the corresponding counter to $tx$ (if it had a larger value).
For presentation purposes, the above description omits an optimization that the \BSTM scheduler uses in cases 1(b) and 2(c), where instead of reducing the counter value, the new task is returned.

\textbf{Optimistic validation.}
An incarnation of transaction might write to a memory location that was previously read by an incarnation of a higher transaction according to the preset serialization order.
This is why in 1(a), when an incarnation finishes, new validation tasks are created for higher transactions.
Importantly, validation tasks are scheduled optimistically, e.g. it is possible to concurrently validate the latest incarnations of transactions $tx_j$, $tx_{j+1}$, $tx_{j+2}$ and $tx_{j+4}$.
Suppose transactions $tx_j$, $tx_{j+1}$ and $tx_{j+4}$ are successfully validated, while the validation of $tx_{j+2}$ fails. 
When threads are available, \BSTM capitalizes by
performing these validations in parallel, allowing it to detect the validation failure of $tx_{j+2}$ faster in the above example (at the expense of a validation of $tx_{j+4}$ that needs to be redone).
Identifying validation failures and aborting incarnations as soon as possible is crucial for the system performance, as any incarnation that reads values written by a incarnation that aborts also needs to be aborted, forming a cascade of aborts.

When an incarnation writes only to a subset of memory locations written by the previously completed incarnation of the same transaction, i.e. case 1(b), \BSTM schedules validation just for the incarnation. This is sufficient due to 2(a), as the whole write-set of the previous incarnation is marked as estimates during the abort. The abort leads to optimistically creating validation tasks for higher transactions in 2(b). Threads that perform these tasks can already detect validation failures due to the \textsc{estimate} markers on memory locations, instead of waiting for a subsequent incarnation to finish.

\textbf{Commit rule.}
In~\cite{gelashvili2022block}, we derive a precise predicate for when transaction $tx_j$ can be considered committed (its roughly when an incarnation is successfully validated after lower transactions $0,\dots,j-1$ have already been committed).
It would be possible to continuously track this predicate, but to reduce the amount of work and synchronization involved, the \BSTM scheduler only checks whether the entire block of transactions can be committed.
This is done by observing that there are no more tasks to perform and at the same time, no threads that are performing any tasks.

\begin{figure}[ht]
\centering
\includegraphics[width=0.8\linewidth]{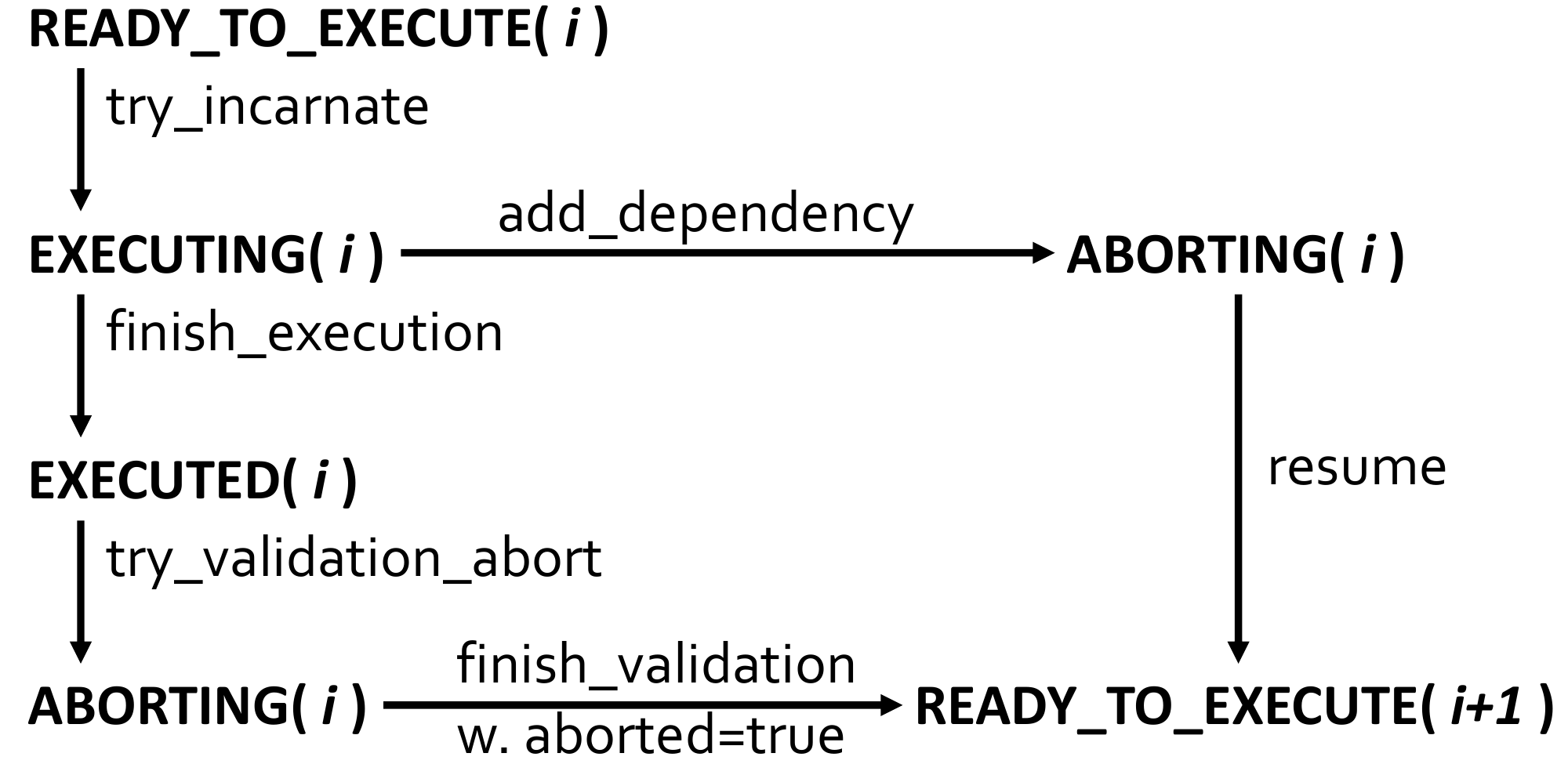}
\caption{Illustration of status transitions}
\label{fig:status}
\vspace{-2mm}
\end{figure}

\begin{algorithm*}[t]
    \caption{Thread logic}
    \label{alg:parexec}
    \small
    \begin{algorithmic}[1]
        
        \Procedure{\rust{run()}}{}
        \label{line:run}
        \State \emph{task} $\gets \bot$
        
        \While{$\neg$\module{Scheduler}\rust{.done()}\label{line:spawnloop}}
            \If{\emph{task} $\neq \bot~\wedge$ \emph{task.kind} $=$ \rust{EXECUTION\_TASK}}
                    \State \emph{task} $\gets$ \rust{try\_execute(\emph{task.version})}\label{line:tryexeccall}\Comment{returns a validation task, or $\bot$}
            \EndIf
            
            \If{\emph{task} $\neq \bot~\wedge$ \emph{task.kind} $=$ \rust{VALIDATION\_TASK}}
                    \State \emph{task} $\gets$                         \rust{needs\_re\-execution(\emph{task.version})}\label{line:reexeccall} \Comment{returns a re-execution task, or $\bot$}
            \EndIf
            
            \If{\emph{task} $= \bot$}
                \State \emph{task} $\gets$ \module{Scheduler}.\rust{next\_task()}
                \label{line:gettaskcall}
            \EndIf
        \EndWhile
        \EndProcedure
        
        
        \Function{\rust{try\_execute}}{\emph{version}} \Comment{returns a validation task, or $\bot$}
            \State (\emph{txn\_idx, incarnation\_number}) $\gets$ \emph{version}
            
            \State \emph{vm\_result} $\gets$ \module{VM}.\rust{execute(\emph{txn\_idx})}\label{line:vmexeccall}
            \Comment{VM does not write to shared memory}
                
            \If{\emph{vm\_result.status} = \rust{READ\_ERROR}}
                \If{$\neg$\module{Scheduler}.\rust{add\_dependency(\emph{txn\_idx, vm\_result.blocking\_txn\_idx})}\label{line:suspendexec}}
                    \State \Return \rust{try\_execute\emph{(version)}}\label{line:retry} \Comment{dependency resolved in the meantime, re-execute} 
                \EndIf
                \State \Return $\bot$
            \Else 
                
            \State \emph{wrote\_new\_location} $\gets$ \module{MVMemory}.\rust{record(\emph{version, vm\_result.read\_set, vm\_result.write\_set})}\label{line:recordcall} 
                
            \State \Return  \module{Scheduler}.\rust{finish\_execution(\emph{txn\_idx, incarnation\_number, wrote\_new\_location})}\label{line:finishexeccall}
            \EndIf             
        \EndFunction
        
        
        \Function{\rust{needs\_reexecution}}{\emph{version}} \Comment{returns a task for re-execution, or $\bot$}
        \label{alg:validate} 
                \State (\emph{txn\_idx, incarnation\_number}) $\gets$ \emph{version}
                
                \State $\emph{read\_set\_valid} \gets $ \module{MVMemory}.\rust{validate\_read\_set(\emph{txn\_idx})} \label{line:allvalidcall}
                
                \State \emph{aborted} $\gets \neg\emph{read\_set\_valid}~\wedge$ \module{Scheduler}.\rust{try\_validation\_abort(\emph{txn\_idx, incarnation\_number})}\label{line:abortcall}
                
                \If{\emph{aborted}}
                    \State \module{MVMemory}.\rust{convert\_writes\_to\_estimates(\emph{txn\_idx})}\label{line:dirtycall}
                \EndIf
                
                \State \Return \module{Scheduler}.\rust{finish\_validation(\emph{txn\_idx, aborted})}\label{line:finvalcall}
		\EndFunction

    \alglinenoNew{counter}
    \alglinenoPush{counter}

    \end{algorithmic}
\end{algorithm*}

\begin{algorithm*}[tbh]
    \caption{The \module{MVMemory} module}
    \label{alg:data}
    \small
    \begin{algorithmic}[1]
    \alglinenoPop{counter}
\Statex \textbf{Atomic Variables:}
    \StateX \emph{data} $\gets$ \rust{Map}, initially empty 
    \Comment{\emph{(location, txn\_idx)} maps to a pair (\emph{incarnation\_number, value}), or to 
    an \rust{ESTIMATE} marker}.

    \StateX \emph{last\_written\_locations} $\gets \rust{Array}(\rust{BLOCK.size()}, \{\})$ 
    \Comment{\emph{txn\_idx} to a set of memory locations written during its last finished execution.}
    
    \StateX \emph{last\_read\_set} $\gets \rust{Array}(\rust{BLOCK.size()}, \{\})$
    \Comment{\emph{txn\_idx} to a set of (\emph{location, version}) 
    pairs per reads in last finished execution.}
    
        \Procedure{\rust{apply\_write\_set}}{\emph{txn\_index, incarnation\_number, write\_set}}
        \For{\textbf{every} \emph{(location, value)} $\in$ \emph{write\_set}}
            \State \emph{data[(location, txn\_idx)]} $\gets$ \emph{(incarnation\_number, value)} \Comment{store in the multi-version data structure}
        \EndFor
        \EndProcedure
        
        \Function{\rust{rcu\_update\_written\_locations}}{\emph{txn\_index, new\_locations}}
        
        \State \emph{prev\_locations} $\gets$ \emph{last\_written\_locations[txn\_idx]} \Comment{loaded atomically (RCU read)}
        
        \For{\textbf{every} \emph{unwritten\_location} $\in $ \emph{prev\_locations} $\setminus$ \emph{new\_locations}}
            \State \emph{data.remove((unwritten\_location, txn\_idx))}\Comment{remove entries that were not overwritten}\label{line:skipremove}
        \EndFor 
        
        \State \emph{last\_written\_locations[txn\_idx]} $\gets$ \emph{new\_locations}
        \Comment{store newly written locations atomically (RCU update)}
        
        \State \Return $\emph{new\_locations} \setminus \emph{prev\_locations} \neq \{\}$ \Comment{was there a write to a location not written the last time}
        \EndFunction
        
        \Function{\rust{record}}{\emph{version, read\_set, write\_set}}
        \State (\emph{txn\_idx, incarnation\_number}) $\gets$ \emph{version}
        
        \State \rust{apply\_write\_set(\emph{txn\_idx, incarnation\_number, write\_set})}
        
        \State \emph{new\_locations} $\gets \{\emph{location} ~|~ (\emph{location}, \star) \in \emph{write\_set}\}$ \Comment{extract locations that were newly written}
        
        \State \emph{wrote\_new\_location} $\gets$ \rust{rcu\_update\_written\_locations(\emph{txn\_idx, new\_locations})} \label{line:uplocscall}
        
        \State \emph{last\_read\_set[txn\_idx]} $\gets$ \emph{read\_set} \Comment{store the read-set atomically (RCU update)} 
        
        \State \Return \emph{wrote\_new\_location}
    \EndFunction
    
    \Procedure{\rust{convert\_writes\_to\_estimates}}{\emph{txn\_idx}}
        \State \emph{prev\_locations} $\gets$ \emph{last\_written\_locations[txn\_idx]} \Comment{loaded atomically (RCU read)}
    
        \For{\textbf{every} \emph{location} $\in$ \emph{prev\_location}}
            \State $\emph{data[(location, txn\_idx)]}\gets$ \rust{ESTIMATE} \label{line:markempty}
            \Comment{entry is guaranteed to exist}
        \EndFor 
    \EndProcedure
    
    \vspace{-3mm}
    \begin{multicols}{2}
    \Function{\rust{read}}{\emph{location, txn\_idx}}
        \State $S \gets \{(\emph{(location, idx)}, \emph{entry}) \in \emph{data} ~|~ \emph{idx} < \emph{txn\_idx}\}$\label{line:upbound}
        
        \If{$S = \{\}$}
            \State \Return \emph{(status $\gets$  \rust{NOT\_FOUND})}
        \EndIf
        \State $\emph{((location, idx)}, \emph{entry}) \gets \argmax_{idx}{S}$ \label{line:argmax} 
        \If{$\emph{entry} = $ \rust{ESTIMATE}}
                \State \Return \emph{(status $\gets$ \rust{READ\_ERROR}, blocking\_txn\_idx $\gets$ idx)} 
         \EndIf
        \State \Return \emph{(status $\gets$ \rust{OK}, version $\gets$ (idx, entry.incarnation\_number), value $\gets$ entry.value)}\label{line:readok}
        
    \EndFunction
    
    \Function{\rust{snapshot()}}{}
        \State \emph{ret} $\gets$ $\{\}$
        
        \For{\textbf{every} \emph{location} $~|~$ ((\emph{location}, $\star$), $\star$) $\in$ \emph{data}}
        
            \State \emph{result} $\gets$ \rust{read(\emph{location, }BLOCK.size())} 
            
            \If{\emph{result.status} $=$ \rust{OK}}
            
                \State \emph{ret} $\gets$ \emph{ret} $\cup$ $\{\emph{location, result.value}\}$
            
            \EndIf
        
        \EndFor
    
        \State \Return \emph{ret} 
    \EndFunction
    \end{multicols}
    
    \vspace{-3mm}
    \Function{\rust{validate\_read\_set}}{\emph{txn\_idx}}
        \State \emph{prior\_reads} $\gets$ \emph{last\_read\_set[txn\_idx]} \Comment{last recorded \emph{read\_set}, loaded atomically via RCU}\label{line:rcurreads}
    
        \For{\textbf{every} \emph{(location, version)} $\in$ \emph{prior\_reads}}\Comment{\emph{version} is $\bot$ when prior read returned \rust{NOT\_FOUND}}
            \State \emph{cur\_read} $\gets$ \rust{read(\emph{location, txn\_idx)}}
                        
            \If{\emph{cur\_read.status} $=$ \rust{READ\_ERROR}}
                \State \Return \emph{false} \Comment{previously read entry from \emph{data}, now \rust{ESTIMATE}}\label{line:mvalidest}
            \EndIf
            
            \If{\emph{cur\_read.status} $=$ \rust{NOT\_FOUND} $\wedge$ \emph{version} $\neq \bot$}
                \State \Return \emph{false} \Comment{previously read entry from \emph{data}, now \rust{NOT\_FOUND}}
            \EndIf
            \If{\emph{cur\_read.status} $=$ \rust{OK} $\wedge$ \emph{cur\_read.version} $\neq$ \emph{version}}
                \State \Return \emph{false}\Comment{read some entry, but not the same as before}
            \EndIf
        \EndFor
        \State \Return \emph{true}
    \EndFunction

    \alglinenoPush{counter}
    \end{algorithmic}
\end{algorithm*}

\begin{algorithm*}[tbh]
    \caption{The \module{VM} module}
    \label{alg:vm}
    \small
    \begin{algorithmic}[1]
    \alglinenoPop{counter}
    \Function{\rust{execute}}{\emph{txn\_id}}
        \State \emph{read\_set} $\gets \{\}$ \Comment{\emph{(location, version)} pairs}
        \State \emph{write\_set} $\gets \{\}$ \Comment{\emph{(location, value)} pairs}
             
        \State \rust{run~transaction}\rust{BLOCK[\emph{txn\_idx}]} \Comment{run transaction, intercept reads and writes}
        
        \Indent
        
        \State ....
        
                \State \textbf{upon writing} \emph{value} \textbf{at a memory} \emph{location}:

        \Indent
            \If{\emph{(location, prev\_value)} $\in$ \emph{write\_set}}
                \State \emph{write\_set} $\gets$ \emph{write\_set}$~\setminus~\{$\emph{(location, prev\_value)}$\}$ \Comment{store the latest value per location}
            \EndIf
        
            \State \emph{write\_set} $\gets$ \emph{write\_set}$~\cup~\{$\emph{(location, value)}$\}$ \Comment{ \textbf{\module{VM}} does not write to \module{MVMemory} or \module{Storage}}
        
        \EndIndent
        
        \State ....
        
        \State \textbf{upon reading a memory} \emph{location}:
        
        \Indent
            \If{\emph{(location, value)} $\in$ \emph{write\_set}}
                \State \module{VM} \textbf{reads} \emph{value}\label{line:inwrite} \Comment{value written by this txn}
            \Else
        
            \State \emph{result} $\gets$ \module{MVMemory}.\rust{read(\emph{location, txn\_idx})}
            
            \If{\emph{result.status} $=$ \rust{NOT\_FOUND}}
                \State \emph{read\_set} $\gets$ \emph{read\_set}$~\cup~\{$\emph{(location, $\bot$)}$\}$\Comment{record version $\bot$ when reading from storage}
                
                \State \textbf{\module{VM} reads} from \module{Storage}
            \ElsIf{\emph{result.status} $=$ \rust{OK}}
                \State \emph{read\_set} $\gets$  \emph{read\_set}$~\cup~\{$\emph{(location, result.version})$\}$
            
                 \State \textbf{\module{VM} reads} \emph{result.value}
            \Else
            
                \State \Return \emph{result} \Comment{return (\rust{READ\_ERROR}, \emph{blocking\_txn\_id}) from the \module{VM}.\rust{execute}}
            
            \EndIf
            \EndIf
        
        \EndIndent
        
        \State ....
        
        \EndIndent
        
              \State \Return \emph{(read\_set, write\_set)} 
    \EndFunction
    \alglinenoPush{counter}
    \end{algorithmic}
\end{algorithm*}

\begin{algorithm*}[tbh]
    \caption{The \module{Scheduler} module, variables, utility APIs and next task logic}
    \label{alg:schvar}
    \small
    \begin{algorithmic}[1]
\alglinenoPop{counter}
\Statex \textbf{Atomic variables:}
	\StateX \emph{execution\_idx} $\gets 0$, \emph{validation\_idx} $\gets 0$, 
	\emph{decrease\_cnt} $\gets 0$,
	\emph{num\_active\_tasks} $\gets 0$,
	\emph{done\_marker} $\gets$ \emph{false}
	
	\Comment{Respectively: An index that tracks the next transaction to try and execute. A similar index for tracking validation. Number of times \StateX \hspace*{2cm} \emph{validation\_idx} or \emph{execution\_idx} was decreased. Number of ongoing validation and execution tasks. Marker for completion.}

	
	
	
    \StateX \emph{txn\_dependency} $\gets~\rust{Array}(\rust{BLOCK.size()},~\rust{mutex}(\{\}))$
    \Comment{\emph{txn\_idx} to a mutex-protected set of dependent transaction indices}

    \StateX \emph{txn\_status} $\gets~\rust{Array}(\rust{BLOCK.size()},~\rust{mutex}((0,\rust{READY\_TO\_EXECUTE})))$
    \Comment{\emph{txn\_idx} to a mutex-protected pair (\emph{incarnation\_number, status}), \Statex \hspace*{7.8cm} where $\emph{status} \in \{ \rust{READY\_TO\_EXECUTE}, \rust{EXECUTING}, \rust{EXECUTED}, \rust{ABORTING}\}$.}
    

    \vspace{-3mm}
    \begin{multicols}{2}
    \Procedure{\rust{decrease\_execution\_idx}}{\emph{target\_idx}}\label{alg:decexproc}
        \State \emph{execution\_idx} $\gets \min(\emph{execution\_idx, target\_idx})$\label{line:decexidx}
        \Comment{atomic}
        
        \State \emph{decrease\_cnt.increment()}\label{line:deccntex}
    \EndProcedure
    
    \Function{\rust{done()}}{}
        \State \Return \emph{done\_marker}\label{line:donemark}
    \EndFunction

	    \Procedure{\rust{decrease\_validation\_idx}}{\emph{target\_idx}}\label{alg:decvalproc} 
        \State \emph{validation\_idx} $\gets \min(\emph{validation\_idx, target\_idx})$
        \Comment{atomic}
        \label{line:decvalidx}
        
        \State \emph{decrease\_cnt.increment()}\label{line:deccntval}
    \EndProcedure

    \Procedure{\rust{check\_done()}}{}
        \State \emph{observed\_cnt} $\gets$ \emph{decrease\_cnt}
    
        \If{$\min(\emph{execution\_idx, validation\_idx})$ $\geq$ \rust{BLOCK.size()} $\wedge$ \\ 
        $\emph{num\_active\_tasks} = 0~\wedge~\emph{observed\_cnt} = \emph{decrease\_cnt}~$}\label{line:donecheck}
                    \State \emph{done\_marker} $\gets$ \emph{true} 
                    \label{line:setdone}
        \EndIf  
    \EndProcedure

    
       \Function{\rust{try\_incarnate}}{\emph{txn\_idx}}
        \If{$\emph{txn\_idx} < \rust{BLOCK.size()}$}\label{line:blockcheckex2}      
            \State \rust{\textbf{with}} \emph{txn\_status[txn\_idx].lock()}\label{line:statlockincarn}
            
            \Indent
            \If{\emph{txn\_status[txn\_idx].status} = \rust{READY\_TO\_EXECUTE}}\label{line:ready}
              \State $\emph{txn\_status[txn\_idx].status} \gets \rust{EXECUTING}$ \label{line:executing} 
            
            \State \Return \emph{(txn\_idx, txn\_status[txn\_idx].incarnation\_number)} 
           \EndIf
           \EndIndent
        \EndIf
        
        \State \emph{num\_active\_tasks.decrement()} \label{line:taskexdec} 
        \State \Return $\bot$
        
    \EndFunction

    \Function{\rust{next\_version\_to\_execute()}}{} 
        \label{alg:nexttoexec} 
        
        \If{\emph{execution\_idx} $\geq$ \rust{BLOCK.size()}}\label{line:blockcheckex1}      
            \State \rust{check\_done()}\label{line:checkcallex}
                
            \State \Return $\bot$
        \EndIf
        
        \State \emph{num\_active\_tasks.increment()} \label{line:taskexinc}
        
        \State \emph{idx\_to\_execute} $\gets$ \emph{execution\_idx.fetch\_and\_increment()}\label{line:fiex} 
        
        \State \Return \rust{try\_incarnate(\emph{idx\_to\_execute})} \label{line:tryinccall}
         
    \EndFunction

    \Function{\rust{next\_version\_to\_validate()}}{}
        \If{$\emph{validation\_idx} \geq \rust{BLOCK.size()}$}\label{line:blockcheckval1}      
            \State \rust{check\_done()}\label{line:checkcallval}

            \State \Return $\bot$
        \EndIf
    
        \State \emph{num\_active\_tasks.increment()} \label{line:taskvalinc}
        
        \State \emph{idx\_to\_validate} $\gets$ \emph{validation\_idx.fetch\_and\_increment()}\label{line:fival} 
        
        \If{$\emph{idx\_to\_validate} < \rust{BLOCK.size()}$}\label{line:blockcheckval2}      
        
        \State \emph{(incarnation\_number, status)} $\gets$ \emph{txn\_status[idx\_to\_validate].lock()} 
        \label{line:statlocknextval}
        
        \If{\emph{status} = \rust{EXECUTED}}\label{line:executed}
            \State \Return \emph{(idx\_to\_validate, incarnation\_number)} 
        \EndIf
 
        \EndIf
        
        \State \emph{num\_active\_tasks.decrement()} \label{line:taskvaldec} 

        \State \Return $\bot$
        
    \EndFunction
    
    \Function{\rust{next\_task()}}{}
        \If{\emph{validation\_idx} $<$ \emph{execution\_idx}}\label{line:idxcompare} 
            \State \emph{version\_to\_validate} $\gets$ \rust{next\_version\_to\_validate()}
            
            \If{\emph{version\_to\_validate} $\neq \bot$}
                \State \Return (\emph{version} $\gets$ \emph{version\_to\_validate},
                \StateX\hspace{18mm}
                \emph{kind} $\gets$ 
                 \rust{VALIDATION\_TASK})
            \EndIf
        \Else
            \State \emph{version\_to\_execute} $\gets$ \rust{next\_version\_to\_execute()}
            
            \If{\emph{version\_to\_execute} $\neq \bot$}
                \State \Return (\emph{version} $\gets$ \emph{version\_to\_execute},
                \StateX\hspace{18mm}
                \emph{kind} $\gets$ \rust{EXECUTION\_TASK})
            \EndIf
        \EndIf
    
        \State \Return $\bot$
    \EndFunction
    \end{multicols}
    \vspace{-3mm}
    
    \alglinenoPush{counter}
    \end{algorithmic}
\end{algorithm*}

\begin{algorithm*}[tbh]
    \caption{The \module{Scheduler} module, dependencies and finish logic}
    \label{alg:schdep}
    \small
    \begin{algorithmic}[1]
\alglinenoPop{counter}    
    \Function{\rust{add\_dependency}}{\emph{txn\_idx}, \emph{blocking\_txn\_idx}} 
        \label{alg:addep} 

        \State \rust{\textbf{with}} \emph{txn\_dependency[blocking\_txn\_idx].lock()}\label{line:lockdep}
        \Indent
            \If{\emph{txn\_status[blocking\_txn\_idx].lock().status} = \rust{EXECUTED}}\label{line:depexcheck} \Comment{thread holds $2$ locks}
                \State \Return \emph{false} \Comment{dependency resolved before locking in~\lineref{line:lockdep}}
            \EndIf
        
            \State \emph{txn\_status[txn\_idx].lock().status()} $\gets$ \rust{ABORTING} \label{line:depaborting} \Comment{previous status must be \rust{EXECUTING}}
            
            \State \emph{txn\_dependency[blocking\_txn\_idx].insert(txn\_idx)}\label{line:adddepid}

        \EndIndent    
        
        \State \emph{num\_active\_tasks.decrement()} \label{line:deptaskdec} \Comment{execution task aborted due to a dependency}
    
        \State \Return \emph{true}
    
    \EndFunction
    
    \Procedure{\rust{set\_ready\_status}}{\emph{txn\_idx}}
        \State \rust{\textbf{with}} \emph{txn\_status[txn\_idx].lock()} \label{line:lockstatready}
        \Indent
                \State \emph{(incarnation\_number, status)} $\gets$ \emph{txn\_status[txn\_idx]} \Comment{\emph{status} must be \rust{ABORTING}}
                
                \State \emph{txn\_status[txn\_idx]} $\gets$ \emph{(incarnation\_number + 1, \rust{READY\_TO\_EXECUTE})} \label{line:setready}
        \EndIndent
    \EndProcedure

    \Procedure{\rust{resume\_dependencies}}{\emph{dependent\_txn\_indices}}
        \For{\textbf{each} \emph{dep\_txn\_idx} $\in$ \emph{dependent\_txn\_indices}}
            \State \rust{set\_ready\_status(\emph{dep\_txn\_idx})}\label{line:depfree}
        \EndFor
        
        \State \emph{min\_dependency\_idx} $\gets$ \emph{min(dependent\_txn\_indices)} \Comment{minimum is $\bot$ if no elements}
        \If{\emph{min\_dependency\_idx} $\neq \bot$}
            \State \rust{decrease\_execution\_idx(\emph{min\_dependency\_idx})}\label{line:depexdec} \Comment{ensure dependent indices get re-executed}
        \EndIf
    \EndProcedure
    
    \Procedure{\rust{finish\_execution}}{\emph{txn\_idx, incarnation\_number, wrote\_new\_path}}
        \State \emph{txn\_status[txn\_idx].lock().status} $\gets$ 
            \rust{EXECUTED}\label{line:finex}    \Comment{status must have been \rust{EXECUTING}}
        
        \State \emph{deps} $\gets$ \emph{txn\_dependency[txn\_idx].lock().swap($\{\})$}\label{line:deplockswap}\Comment{swap out the set of dependent transaction indices}
        
        \State \rust{resume\_dependencies(\emph{deps})}\label{line:resumecall}
        
        \If{\emph{validation\_idx} $>$ \emph{txn\_idx}} \Comment{otherwise index already small enough} \label{line:valcompfin}
            \If{\emph{wrote\_new\_path}}
                \State \rust{decrease\_validation\_idx(\emph{txn\_idx})}\Comment{schedule validation for \emph{txn\_idx} and higher txns}\label{line:decvalcall2}
            \Else
                \State \Return (\emph{version} $\gets$ \emph{(txn\_idx, incarnation\_number), kind} $\gets$ \rust{VALIDATION\_TASK}) \label{line:retvaltask}
        \EndIf
        \EndIf
        
        \State \emph{num\_active\_tasks.decrement()} \label{line:extaskdec} 
        \State \Return $\bot$ \Comment{no task returned to the caller}
    \EndProcedure
    
    \Function{\rust{try\_validation\_abort}}{\emph{txn\_idx, incarnation\_number}}
        \State \rust{\textbf{with}} \emph{txn\_status[txn\_idx].lock()}
            \Indent
            \If{\emph{txn\_status[txn\_idx]} = \emph{(incarnation\_number, \rust{EXECUTED})}}\label{line:canabort}

              \State $\emph{txn\_status[txn\_idx].status} \gets \rust{ABORTING}$ \Comment{thread changes status, starts aborting} \label{line:valaborting}
            
            \State \Return \emph{true}
           \EndIf
           \EndIndent
        \State \Return \emph{false}
    \EndFunction

    \Procedure{\rust{finish\_validation}}{\emph{txn\_idx, aborted}}
        \If{\emph{aborted}}    
            \State \rust{set\_ready\_status(\emph{txn\_idx})}\label{line:readycallval}
            
            \State \rust{decrease\_validation\_idx(\emph{txn\_idx}} $+~1$\rust{)}
            \label{line:decvalcall1}\Comment{schedule validation for higher transactions}
            
            \If{\emph{execution\_idx} $>$ \emph{txn\_idx}}\label{line:exidxfincheck} \Comment{otherwise index already small enough}
                \State \emph{new\_version} $\gets$ \rust{try\_incarnate(\emph{txn\_idx})}\label{line:finincarnate}
                
                \If{\emph{new\_version} $\neq \bot$}\label{line:finincarnver}
                     \State \Return (\emph{new\_version}, \emph{kind} $\gets$ \rust{EXECUTION\_TASK}) \label{line:finvalextask} \Comment{return re-execution task to the caller}
                \EndIf
            \EndIf
        \EndIf
        
        \State \emph{num\_active\_tasks.decrement()} \Comment{done with validation task}\label{line:valtaskdec}
        \State \Return $\bot$ \Comment{no task returned to the caller}
    \EndProcedure
\alglinenoPush{counter}
    \end{algorithmic}
\end{algorithm*}

\vspace{-2mm}
\section{\BSTM Detailed Description}
\label{sec:details}
In this section, we describe \BSTM.
%
Upon spawning, threads perform the \rust{run()} procedure in~\lineref{line:run}. Our pseudo-code is divided into several modules that the threads use. The \module{Scheduler} module contains the shared variables and logic used to dispatch execution and validation tasks. The \module{MVMemory}
module contains shared memory in a form of a multi-version data-structure for values written and read by different transactions in \BSTM. Finally, the \module{VM} module describes how reads and writes are instrumented during transaction execution.

\BSTM finishes when all threads join after returning from the \rust{run()} invocation. At this point, the output of \BSTM can be obtained by calling the \module{MVMemory}.\rust{snapshot()} function that returns the final values for all affected memory locations.
This function can be easily parallelized and the output can be persisted to main storage (abstracted as a \module{Storage} module), but these aspects are out of the scope here.

\subsection{High-Level Thread Logic}
\label{sec:highthread}
We start by the high-level logic described in~\algorithmref{alg:parexec}.
The \rust{run()} procedure interfaces with the \module{Scheduler} module and consists of a loop that lets the invoking thread continuously perform available validation and execution tasks. 
The thread looks for a new task in~\lineref{line:gettaskcall}, and dispatches a proper handler based on its kind, i.e. function \rust{try\_execute} in~\lineref{line:tryexeccall} for an \rust{EXECUTION\_TASK} and function \rust{needs\_reexecution} in~\lineref{line:reexeccall} for a \rust{VALIDATION\_TASK} (since, as discussed in~\sectionref{sec:oview}, a successful validation does not change state, while failed validation implies that the transaction requires re-execution). 
Both of this functions take a transaction version (transaction index and incarnation number) as an input.
A \rust{try\_execute} function invocation may return a new validation task back to the caller, and a \rust{needs\_reexecution} function invocation may return a new execution task.

\vspace{-1mm}
\subsubsection{Execution Tasks}
An execution task is processed using the \rust{try\_execute} procedure. 
First, a \module{VM}.\rust{execute} function is invoked in~\lineref{line:vmexeccall}.
As discussed in~\sectionref{sec:vm}, by the \module{VM} design, this function reads from memory (\module{MVMemory} data-structure and the main \module{Storage}), but never modifies any state while being performed. Instead, a successful \module{VM} execution returns a \emph{write-set}, consisting of memory locations and their updated values, which are applied to \module{MVMemory} by the \rust{record} function invocation in~\lineref{line:recordcall}.
In \BSTM, \module{VM}.\rust{execute} also captures and returns a \emph{read-set}, containing all memory locations read during the incarnation, each associated with whether a value was read from \module{MVMemory} or \module{Storage}, and in the former case, the version of the transaction execution that previously wrote the value. 
The read-set is also passed to the \module{MVMemory}.\rust{record} call in~\lineref{line:recordcall} and stored in \module{MVMemory} for later validation purposes. 

Every \module{MVMemory}.\rust{record} invocation returns an indicator whether a write occurred to a memory location not written to by the previous incarnation of the same transaction.
As discussed in~\sectionref{sec:oview}, in \BSTM this indicator determines whether the higher transactions (than the transaction that just finished execution, in the preset serialization order)  require further validation.
\module{Scheduler}.\rust{finish\_execution} in \lineref{line:finishexeccall} schedules the required validation tasks.
When a new location is not written, \emph{wrote\_new\_location} variable is set to \emph{false} and it suffices to only validate the transaction itself.
In this case, due to an internal performance optimization, the \module{Scheduler} module sometimes returns this validation task back to the caller from the \rust{finish\_execution} invocation.

The \module{VM} execution of transaction $tx_j$ may observe a read dependency on a lower transaction $tx_k$ in the preset order, $k < j$. As discussed in~\sectionref{sec:oview}, this happens when the last incarnation of $tx_k$ wrote to a memory location that $tx_j$ reads, but when the incarnation of $tx_k$ aborted before the read by $tx_j$.
In this case, the index $k$ of the blocking transaction is returned as \emph{vm\_result.blocking\_txn\_idx}, a part of the output in~\lineref{line:vmexeccall}. 
In order to re-schedule the execution task for $tx_j$ for after the blocking transaction $tx_k$ finishes its next incarnation,~\module{Scheduler}.\rust{add\_dependency} is called  in~\lineref{line:suspendexec}. This function returns \emph{false} if it encounters a race condition when $tx_k$ gets re-executed before the dependency can be added. The execution task is then retried immediately in~\lineref{line:retry}.

\subsubsection{Validation Tasks}
\label{sec:valtask}
A \rust{validate\_read\_set} call in~\lineref{line:allvalidcall} obtains 
the last read-set recorded by an execution of \emph{txn\_idx} and checks that re-reading each memory location in the read-set still yields the same values.
To be more precise, for every value that was read, the read-set stores a read descriptor.
This descriptor contains the version of the transaction (during the execution of which the value was written), or $\bot$ if the value was read from storage (i.e. not written by a smaller transaction).
The incarnation numbers are monotonically increasing, so it is sufficient to validate the read-set by comparing the corresponding descriptors.

If validation fails, \rust{try\_validation\_abort} on \module{Scheduler} is called in~\lineref{line:abortcall}, which returns an indicator of whether the abort was successful.
\module{Scheduler} ensures that only one failing validation per version may lead to a successful abort.
Hence, if \rust{abort\_validation} returns \emph{false}, then the incarnation was already aborted.
If the abort was successful, then \rust{convert\_writes\_to\_estimates(\emph{txn\_idx})} in~\lineref{line:dirtycall} is called, which replaces the write-set of the aborted version in the shared memory data-structure with special \rust{ESTIMATE} markers.
A successful abort leads to scheduling the transaction for re-execution and the higher transactions for validation during the \module{Scheduler}.\rust{finish\_validation} call in~\lineref{line:finvalcall}.
Sometimes, (as an optimization), the re-execution task is returned (that proceeds to return the new version from \rust{needs\_reexecution} and then in~\lineref{line:tryexeccall} become the only thread to execute the next incarnation of the transaction).

\subsection{Multi-Version Memory}
The \module{MVMemory} module (\algorithmref{alg:data}) describes the shared memory data-structure in \BSTM. It is called \emph{multi-version} because it stores multiple writes for each memory location, along with a value and an associated version of a corresponding transaction.
In the pseudo-code, we represent the main data-structure, called \emph{data}, with an abstract map interface, mapping (\emph{location}, \emph{txn\_idx}) pairs to the corresponding entries, which are (\emph{incarnation\_number}, \emph{value}) pairs.
In order to support a read of memory \emph{location} by transaction $tx_j$, \emph{data} provides an interface that returns an entry written at location by the transaction with the highest index $i$ such that $i < j$
This functionality is used 
in~\lineref{line:upbound} and~\lineref{line:argmax}. For clarity of presentation, our pseudo-code focuses on the abstract functionality of the map, while standard concurrent data-structure design techniques can be used for an efficient implementation (discussed in~\sectionref{sec:ev}).

For every transaction, \module{MVMemory} stores a set of memory locations in the \emph{last\_written\_locations} array and a set of \emph{(location, version)} pairs in the \emph{last\_read\_set} array. We assume that these sets are loaded and stored atomically, which can be accomplished by storing a pointer to the set and accessing the pointer atomically, i.e. via the read-copy-update~\cite{rcu}.

\textbf{Recording.}
The \rust{record} function takes a transaction version along with the read-set and the write-set (resulting from the execution of the version). 
The write-set consists of (memory location, value) pairs that are applied to the \emph{data} map by \rust{apply\_write\_set} procedure invocation.
The invocation of \rust{rcu\_update\_written\_locations} that follows in~\lineref{line:uplocscall} updates \emph{last\_written\_locations} and also removes (in~\lineref{line:skipremove}) from the \emph{data} map all entries at memory locations that were not overwritten by the latest write-set of the transaction (i.e. locations in the \emph{last\_written\_locations} before, but not after the update). This function also determines and returns whether a new memory location was written (i.e. in \emph{last\_written\_locations} after, but not before the update). 
This indicator is stored in \emph{wrote\_new\_location} variable and returned from the \rust{record} function. Before returning, the read-set of the transaction is stored in the \emph{last\_read\_set} array via an RCU pointer update.


The \rust{convert\_writes\_to\_estimates} procedure, called during a transaction abort, iterates over \emph{last\_written\_locations} of the transaction, and replaces each stored (\emph{incarnation\_number}, \emph{value}) pair with a special \rust{ESTIMATE} marker.
It ensures that validations fail for higher transactions if they  
have read the data written by the aborted incarnation. 
While removing the entries can also accomplish this, the \rust{ESTIMATE} marker also serves as a ``write estimate" for the next incarnation of this transaction. Any transaction that observes an \rust{ESTIMATE} of transaction $tx$ when reading during a speculative execution, waits for the dependency to resolve ($tx$ to be re-executed), as opposed to ignoring the \rust{ESTIMATE} and likely aborting if $tx$'s next incarnation again writes to the same memory location.

\textbf{Reads.}
The \module{MVMemory}.\rust{read} function takes a memory location and a transaction index \emph{txn\_idx} as its input parameters. First, it looks for the highest transaction index, \emph{idx}, among transactions lower than \emph{txn\_idx} that have written to this memory location (\lineref{line:upbound} and~\lineref{line:argmax}). Based on the fixed serialization order of transactions in the block, this is the best guess for reading speculatively (writes by transactions lower than \emph{idx} are overwritten by \emph{idx}, and the speculative premise is that the transactions between \emph{idx} and \emph{txn\_idx} do not write to the same memory location).
The value written by transaction \emph{idx} is returned in~\lineref{line:readok}, alongside with the full version (i.e. \emph{idx} and the incarnation number) and an \rust{OK} status.  
However, if the entry corresponding to transaction \emph{idx} is an \rust{ESTIMATE} marker, then the \rust{read} returns an \rust{READ\_ERROR} status and \emph{idx} as a blocking transaction index. This is an indication for the caller to postpone the execution of transaction \emph{txn\_idx} until the next incarnation of the blocking transaction \emph{idx} completes. Essentially, at this point, it is estimated that transaction \emph{idx} will perform a write that is relevant for the correct execution of transaction \emph{txn\_idx}.

When no lower transaction has written to the memory location, a \rust{read} returns a \rust{NOT\_FOUND} status, implying that the value cannot be obtained from the previous transactions in the block. As we will describe shortly, the caller can then complete the speculative read by reading from storage.

The \rust{validate\_read\_set} function loads (via RCU) 
the most recently recorded read-set from the transaction's execution in~\lineref{line:rcurreads}. The function calls \rust{read} for each location and checks observed status and version against the read-set (recall that version $\bot$ in the read-set means that the corresponding prior read returned \rust{NOT\_FOUND} status, i.e. it read a value from \module{Storage}). As we saw in~\sectionref{sec:valtask}, \rust{validate\_read\_set} function is invoked during validation in~\lineref{line:allvalidcall}, at which point the incarnation that is being validated is already executed and has recorded the read-set. However, if the thread performing a validation task for incarnation $i$ of a transaction is slow, it is possible that \rust{validate\_read\_set} function invocation observes a read-set recorded by a later (i.e. $>i$) incarnation. In this case, incarnation $i$ is guaranteed to be already aborted (else higher incarnations would never start), and the validation task will have no effect on the system regardless of the outcome (only validations that successfully abort affect the state and each incarnation can be aborted at most once).

The \rust{snapshot} function is called after \BSTM finishes, and returns the value written by the highest transaction for every location that was written to by some transaction.

\subsubsection{\module{VM} execution}
\label{sec:vm}
In~\algorithmref{alg:vm} we describe how reads and writes are handled in \BSTM by the \module{VM}.\rust{execute} function (invoked while performing an execution task, in~\lineref{line:vmexeccall}). 
This function tracks and returns the transaction's read- and write-sets, both initialized to empty.
When a transaction attempts to write a value to a location, the (location, value) pair is added to the write-set, possibly replacing a pair with a prior value (if it is not the first time the transaction wrote to this location during the execution). 

When a transaction attempts to read a location, if the location is already in the write-set then the \module{VM} reads the corresponding value (that the transaction itself wrote) in~\lineref{line:inwrite}.
Otherwise, \module{MVMemory}.\rust{read} is performed. If it returns \\ \rust{NOT\_FOUND}, then \rust{VM} reads the value directly from storage (abstracted as a \module{Storage} module that contains values preceding the block execution) and records \emph{(location, $\bot$)} in the read-set. If \module{MVMemory}.\rust{read} returns  \rust{READ\_ERROR}, then \module{VM} execution stops and returns the error and the blocking transaction index (for the dependency) to the caller. If it returns \rust{OK}, then  \module{VM} reads the resulting value from \module{MVMemory} and records the location and version pair in the read-set.

Note that for simplicity of presentation, if the transaction reads the same location more than once, the pseudo-code repeats the \rust{read} and makes separate record in the read-set. Even if reading the same location results in reading different values, \BSTM algorithm maintains correctness because all reads are eventually validated and the \module{VM} captures the errors that may arise due to any opacity violations.

\vspace{-2mm}
\subsection{Scheduling}
The \module{Scheduler} module contains the necessary state and synchronization logic for managing the execution and validation tasks. For each transaction in a block, the \emph{txn\_status} array contains the most up-to-date incarnation number (initially $0$) and the status of this incarnation, which can be one of \rust{READY\_TO\_EXECUTE} (initial value), \rust{EXECUTING}, \rust{EXECUTED} and \rust{ABORTING}. The entries of the \emph{txn\_status} array are protected by a lock to provide atomicity. 

Status transitions are illustrated in~\figureref{fig:status}.
The thread that changes the status from \rust{READY\_TO\_EXECUTE} to \rust{EXECUTING} in~\lineref{line:executing} when incarnation number is $i$ performs incarnation $i$ of the transaction.
The status never becomes \\ \rust{READY\_TO\_EXECUTE(i)} again, guaranteeing that no incarnation is performed more than once. 
Afterwards, this thread sets the status to \rust{EXECUTED(i)} in~\lineref{line:finex}.
Similarly, only the thread that changes the status from \rust{EXECUTED(i)} to \rust{ABORTING(i)} returns \emph{true} from \rust{try\_validation\_abort} for incarnation $i$. After performing the steps associated with a successful abort, as discussed in~\sectionref{sec:valtask}, this thread then updates the status to\\ \rust{READY\_TO\_EXECUTE(i+1)} in~\lineref{line:setready}.
This indicates that an execution task for incarnation $i+1$ is ready to be created.

When incarnation $i$ of transaction $tx_k$ aborts because of a read dependency on transaction $tx_j$ ($j < k$ in the preset serialization order), the status of $tx_k$ is updated to \rust{ABORTING(i)} in~\lineref{line:depaborting}. 
The corresponding \rust{add\_dependency(\emph{k, j})} invocation returns \emph{true} and \BSTM guarantees that some thread will subsequently finish executing transaction $tx_j$ and resolve $tx_k$'s dependency in~\lineref{line:setready} (called from~\lineref{line:depfree}) by setting its status to \rust{READY\_TO\_EXECUTE(i+1)}. 

The \emph{txn\_dependency} array is used to track transaction dependencies.
In the above example, when transaction $tx_k$ reads an estimate of transaction $tx_j$ and calls \rust{add\_dependency(\emph{k, j})} (that returns \emph{true}), $k$ is added to \emph{txn\_dependency[j]} in~\lineref{line:adddepid}.
Our pseudo-code explicitly describes lock-based synchronization for the dependencies stored in the \emph{txn\_dependency} array. This is to demonstrate the handling of a race between the \rust{add\_dependency} function of $tx_k$ and the \rust{finish\_execution} procedure of $tx_j$ (in particular, to guarantee that transaction $tx_j$ will always clear its dependencies in~\lineref{line:deplockswap}).
The problematic scenario could arise if after $tx_k$ observed the read dependency, transaction $tx_j$ raced to \rust{finish\_execution} and cleared its dependencies.
However, due to the check in~\lineref{line:depexcheck}, dependency will not be added and the \rust{add\_dependency} invocation will return \emph{false}.
Then, the status of $tx_k$ would remain \rust{EXECUTING} and the caller would immediately re-attempt the execution task of $tx_k$, incarnation $i$, in~\lineref{line:retry}.


\textbf{Managing Tasks.}
\BSTM scheduler maintains \emph{execution\_idx} and \emph{validation\_idx} atomic counters.
Together, one can view the status array and the validation (or execution) index counter as a counting-based implementation of an ordered set abstraction for selecting lowest-indexed available validation (or execution) task.

The \emph{validation\_idx} counter tracks the index of the next transaction to be validated.
A thread picks an index in~\lineref{line:fival} in the  \rust{next\_version\_to\_validate} function by performing the \emph{fetch\_and\_increment} instruction on the \emph{validation\_idx}.
It then checks if the transaction with the corresponding index is ready to be validated (i.e. the status is \rust{EXECUTED}), and if it is, determines the latest incarnation number.
A similar \emph{execution\_idx} counter is used in combination with the status array to manage execution tasks. 
In the \rust{next\_version\_to\_execute} function, a thread picks an index by \emph{fetch\_and\_increment}-ing in~\lineref{line:fiex}, then invokes the \rust{try\_incarnate} function.
Only if the transaction is in a \rust{READY\_TO\_EXECUTE} state, this function will set the status to \rust{EXECUTING} and return the corresponding version for execution.

When transaction status is updated to \rust{READY\_TO\_EXECUTE}, \BSTM ensures that the corresponding execution task eventually gets created. In the \rust{resume\_dependencies} procedure, the execution index is reduced by the call in~\lineref{line:depexdec} to be no higher than indices of all transactions that had a dependency resolved. In \rust{finish\_validation} function after a successful abort, however, there may be a single re-execution task (unless the task was already claimed by another thread after the status was set, something that is checked in~\lineref{line:finincarnver}). As an optimization, instead of reducing \emph{execution\_idx}, the execution task is sometimes returned to the caller in~\lineref{line:finvalextask}.

Similarly, if a validation of transaction $tx_k$ was successfully aborted, then \BSTM ensures, in the \rust{finish\_validation} function (in~\lineref{line:decvalcall1}), that $\emph{validation\_idx} \leq k$.
In addition, in the \rust{finish\_execution} function of transaction $tx_k$, \BSTM invokes \rust{decrease\_validation\_idx} in~\lineref{line:decvalcall2} if a new memory location was written by the associated incarnation. Otherwise, only a validation task for $tx_k$ is created that may be returned to the caller.

Finally, the \rust{next\_task} function decides whether to obtain a version to execute or version to validate based on a simple heuristic, by comparing the two indices in~\lineref{line:idxcompare}.

\textbf{Detecting Completion.}
The \module{Scheduler} provides a mechanism for the threads to detect when all execution and validation tasks are completed. This is not trivial because individual threads might obtain no available tasks from the \rust{next\_task} function, but more execution and validation tasks could still be created later, e.g. if a validation task that is being performed by another thread fails.

\BSTM implements a \rust{check\_done} procedure that determines when all work is completed and the threads can safely return. In this case, a \emph{done\_marker} is set to \emph{true}, providing a cheap way for all threads to exit their main loops in~\lineref{line:spawnloop}.
Threads invoke a \rust{check\_done} procedure in~\lineref{line:checkcallex} and~\lineref{line:checkcallval}, when observing an execution or validation index that is already $\geq \rust{BLOCK.size()}$.
In the following, we explain the logic behind \rust{check\_done}.

A straw man approach would be to check that both execution and validation indices are at least as large as the \rust{BLOCK.size()}. The first problem with this approach is that it does not consider when the execution and validation tasks actually finish. For example, the \emph{validation\_idx} may be incremented in~\lineref{line:fival} and become \rust{BLOCK.size()}, but it would be incorrect for the threads to return, as the corresponding validation task of transaction $\rust{BLOCK.size()} -1$ may still fail. 
To overcome this problem, \BSTM utilizes the \emph{num\_active\_tasks} atomic counter to track the number of ongoing execution and validation tasks. Then, in addition to the indices, the scheduler also checks whether $\emph{num\_active\_tasks} = 0$ in~\lineref{line:donecheck}.

The \emph{num\_active\_tasks} counter is incremented in~\lineref{line:taskexinc} and~\lineref{line:taskvalinc}, right before \emph{execution\_idx} and \emph{validation\_idx} are \emph{fetch-and-increment}-ed, respectively. The 
\emph{num\_active\_tasks} is decremented if no task corresponding to the fetched index is created (\lineref{line:taskexdec} and~\lineref{line:taskvaldec}), or after the tasks finish (\lineref{line:extaskdec} and~\lineref{line:valtaskdec}). As an optimization, when \rust{finish\_execution} or \rust{finish\_validation} functions return a new task to the caller, \emph{num\_active\_tasks} is left unchanged (instead of incrementing and decrementing that cancel out). 

The second challenge is that \emph{validation\_idx}, \emph{execution\_idx} and \emph{num\_active\_tasks} are separate counters, e.g. it is possible to read that \emph{validation\_idx} has value \rust{BLOCK.size()}, then read that \emph{num\_active\_tasks} has value $0$, without these variables simultaneously holding the respective values. \BSTM handles this by another counter, \emph{decrease\_cnt}, incremented in \rust{decrease\_execution\_idx} and \\ \rust{decrease\_validation\_idx} procedures (\lineref{line:deccntex},~\lineref{line:deccntval}). By reading \emph{decrease\_cnt} twice in \rust{check\_done}, it is possible to detect if validation or execution index decreases from their observed values when \emph{num\_active\_tasks} is read to be $0$.

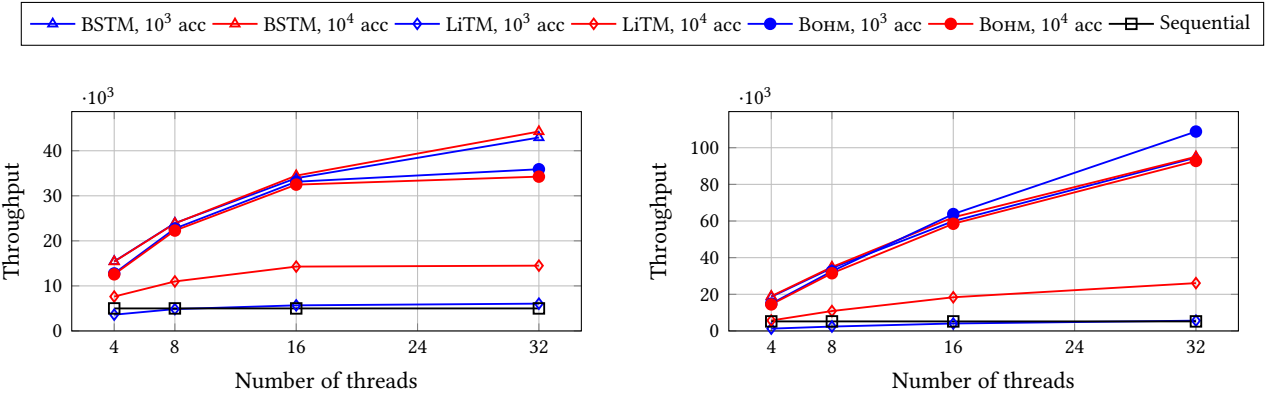
\begin{figure*}[t]
\begin{minipage}[t]{0.47\linewidth}
    \centering
    \pgfplotsset{footnotesize,height=4.5cm, width=\linewidth}
    \begin{tikzpicture}
    \begin{axis}[
        scaled y ticks=base 10:-3,
        legend columns=10,
        legend style={font=\footnotesize, at={(-0.1,1.5)},anchor=north west,legend columns=2},
        xlabel={Number of threads},
        ylabel= {Throughput},
        grid=major,
        xtick={4, 8, 16, 24, 32},
        ymin=0,
        ]
        \addplot [line width=0.25mm, color=blue, mark=triangle] table [x=threads, y=bstm1k,  col sep=comma] {plots/batch-1k};
        \addplot [line width=0.25mm, color=red, mark=triangle] table [x=threads, y=bstm10k,  col sep=comma] {plots/batch-1k};
        \addplot [line width=0.25mm, color=blue, mark=diamond] table [x=threads, y=litm1k,  col sep=comma] {plots/batch-1k};
        \addplot [line width=0.25mm, color=red, mark=diamond] table [x=threads, y=litm10k,  col sep=comma] {plots/batch-1k};
        \addplot [line width=0.25mm, color=blue, mark=*] table [x=threads, y=bohm1kexe,  col sep=comma] {plots/batch-1k-bohm};
        \addplot [line width=0.25mm, color=red, mark=*] table [x=threads, y=bohm10kexe,  col sep=comma] {plots/batch-1k-bohm};
        \addplot [line width=0.25mm, black, mark=square] [error bars/.cd, y explicit,y dir=both,] table [x=threads, y=seq,  col sep=comma] {plots/batch-1k};
        \addlegendentry{BSTM, $10^3$ acc}
        \addlegendentry{BSTM, $10^4$ acc}
        \addlegendentry{LiTM, $10^3$ acc}
        \addlegendentry{LiTM, $10^4$ acc}
        \addlegendentry{\textsc{Bohm}, $10^3$ acc}
        \addlegendentry{\textsc{Bohm}, $10^4$ acc}
        \addlegendentry{Sequential}
    \end{axis}
    \end{tikzpicture}
    \label{fig:all-1k-p2p}
\end{minipage}
\begin{minipage}[t]{0.47\linewidth}
    \centering
    \pgfplotsset{footnotesize,height=4.5cm, width=\linewidth}
    \begin{tikzpicture}
    \begin{axis}[
        scaled y ticks=base 10:-3,
        legend columns=5,
        xlabel={Number of threads},
        ylabel= {Throughput},
        grid=major,
        xtick={4, 8, 16, 24, 32},
        ymin=0,
        ]
        \addplot [line width=0.25mm, color=blue, mark=triangle] table [x=threads, y=bstm1k,  col sep=comma] {plots/batch-10k};
        \addplot [line width=0.25mm, color=red, mark=triangle] table [x=threads, y=bstm10k,  col sep=comma] {plots/batch-10k};
        \addplot [line width=0.25mm, color=blue, mark=diamond] table [x=threads, y=litm1k,  col sep=comma] {plots/batch-10k};
        \addplot [line width=0.25mm, color=red, mark=diamond] table [x=threads, y=litm10k,  col sep=comma] {plots/batch-10k};
        \addplot [line width=0.25mm, color=blue, mark=*] table [x=threads, y=bohm1kexe,  col sep=comma] {plots/batch-10k-bohm};
        \addplot [line width=0.25mm, color=red, mark=*] table [x=threads, y=bohm10kexe,  col sep=comma] {plots/batch-10k-bohm};
        \addplot [line width=0.25mm, black, mark=square] [error bars/.cd, y explicit,y dir=both,] table [x=threads, y=seq,  col sep=comma] {plots/batch-10k};
    \end{axis}
    \end{tikzpicture}
    \label{fig:all-10k-p2p}
\end{minipage}
\vspace{-3mm}
 \caption{ Comparison of BSTM, LiTM, \textsc{Bohm} and sequential execution for block size $10^3$ (left) and $10^4$ (right). \textsc{Bohm} is provided with perfect write estimates. Diem p2p txns.
 }
 \vspace{-3mm}
 \label{fig:all-diem}
\end{figure*}

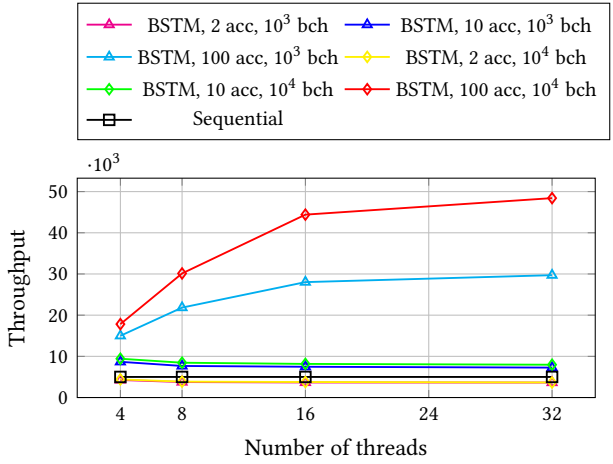
\begin{figure}[t]
    \centering
    \pgfplotsset{footnotesize,height=4.5cm, width=\linewidth}
    \begin{tikzpicture}
    \begin{axis}[
        scaled y ticks=base 10:-3,
        legend columns=2,
        legend style={font=\footnotesize, at={(0.0,1.8)},anchor=north west,legend columns=2},
        xlabel={Number of threads},
        ylabel= {Throughput},
        grid=major,
        xtick={4, 8, 16, 24, 32},
        ymin = 0,
        ]
        \addplot [line width=0.25mm, color=magenta, mark=triangle] table [x=threads, y=bstm2,  col sep=comma] {plots/batch-1k};
        \addplot [line width=0.25mm, color=blue, mark=triangle] table [x=threads, y=bstm10,  col sep=comma] {plots/batch-1k};
        \addplot [line width=0.25mm, color=cyan, mark=triangle] table [x=threads, y=bstm100,  col sep=comma] {plots/batch-1k};
        \addplot [line width=0.25mm, color=yellow, mark=diamond] table [x=threads, y=bstm2,  col sep=comma] {plots/batch-10k};
        \addplot [line width=0.25mm, color=green, mark=diamond] table [x=threads, y=bstm10,  col sep=comma] {plots/batch-10k};
        \addplot [line width=0.25mm, color=red, mark=diamond] table [x=threads, y=bstm100,  col sep=comma] {plots/batch-10k};
        \addplot [line width=0.25mm, black, mark=square] [error bars/.cd, y explicit,y dir=both,] table [x=threads, y=seq,  col sep=comma] {plots/batch-1k};
        \addlegendentry{BSTM, 2 acc, $10^3$ bch}
        \addlegendentry{BSTM, 10 acc, $10^3$ bch}
        \addlegendentry{BSTM, 100 acc, $10^3$ bch}
        \addlegendentry{BSTM, 2 acc, $10^4$ bch}
        \addlegendentry{BSTM, 10 acc, $10^4$ bch}
        \addlegendentry{BSTM, 100 acc, $10^4$ bch}
        \addlegendentry{Sequential}
    \end{axis}
    \end{tikzpicture}
    \caption{ Comparison of BSTM and sequential execution for block size $10^3$ and $10^4$, account sizes 2, 10 and 100. Diem p2p transactions.}
    \label{fig:bstm-high-diem}
    \vspace{-3mm}
\end{figure}

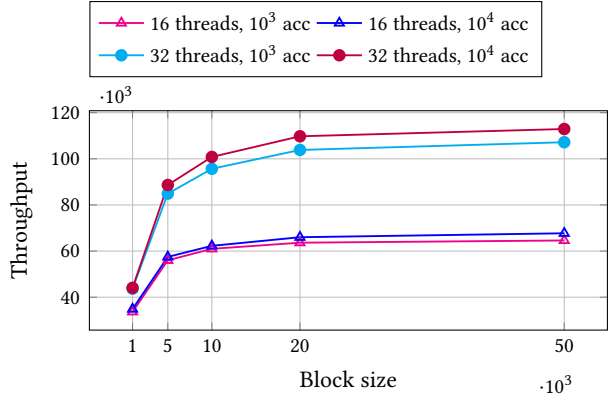
\begin{figure}[t]
    \centering
    \pgfplotsset{footnotesize,height=4.5cm, width=\linewidth}
    \begin{tikzpicture}
    \begin{axis}[
        scaled ticks=base 10:-3,
        legend columns=2,
        legend style={font=\footnotesize, at={(0,1.5)},anchor=north west,legend columns=2},
        xlabel={Block size},
        ylabel= {Throughput},
        grid=major,
        xtick={1000, 5000, 10000, 20000, 50000},
        ]
        \addplot [line width=0.25mm, color=magenta, mark=triangle] table [x=batches, y=bstm-16-1k,  col sep=comma] {plots/batches};
        \addplot [line width=0.25mm, color=blue, mark=triangle] table [x=batches, y=bstm-16-10k,  col sep=comma] {plots/batches};
        \addplot [line width=0.25mm, color=cyan, mark=*] table [x=batches, y=bstm-32-1k,  col sep=comma] {plots/batches};
        \addplot [line width=0.25mm, color=purple, mark=*] table [x=batches, y=bstm-32-10k,  col sep=comma] {plots/batches};
        \addlegendentry{16 threads, $10^3$ acc}
        \addlegendentry{16 threads, $10^4$ acc}
        \addlegendentry{32 threads, $10^3$ acc}
        \addlegendentry{32 threads, $10^4$ acc}
    \end{axis}
    \end{tikzpicture}
    \caption{ Throughput of BSTM for various block sizes. Diem p2p transactions.}
    \label{fig:tps-diem}
    \vspace{-3mm}
\end{figure}

\begin{figure*}[t]
\begin{minipage}[b]{0.45\linewidth}
    \centering
    \pgfplotsset{footnotesize,height=4.5cm, width=\linewidth}
    \begin{tikzpicture}
    \begin{axis}[
        scaled y ticks=base 10:-3,
        legend columns=4,
        legend style={font=\footnotesize, at={(0.6,1.5)},anchor=north west},
        xlabel={Number of threads},
        ylabel= {Throughput},
        grid=major,
        xtick={4, 8, 16, 24, 32},
        ymin=0,
        ]
        %
        %
        \addplot [line width=0.25mm, color=blue, mark=triangle] table [x=threads, y=bstm1k,  col sep=comma] {plots/aptos-batch-1k};
        \addplot [line width=0.25mm, color=red, mark=triangle] table [x=threads, y=bstm10k,  col sep=comma] {plots/aptos-batch-1k};
        %
        %
        %
        \addplot [line width=0.25mm, black, mark=square] [error bars/.cd, y explicit,y dir=both,] table [x=threads, y=seq,  col sep=comma] {plots/aptos-batch-1k};
        %
        %
        \addlegendentry{BSTM, $10^3$ acc}
        \addlegendentry{BSTM, $10^4$ acc}
        \addlegendentry{Sequential}
    \end{axis}
    \end{tikzpicture}
\end{minipage}
\hspace{4ex}
\begin{minipage}[b]{0.45\linewidth}
    \centering
    \pgfplotsset{footnotesize,height=4.5cm, width=\linewidth}
    \begin{tikzpicture}
    \begin{axis}[
        scaled y ticks=base 10:-3,
        xlabel={Number of threads},
        ylabel= {Throughput},
        grid=major,
        xtick={4, 8, 16, 24, 32},
        ymin=0,
        ]
        %
        %
        \addplot [line width=0.25mm, color=blue, mark=triangle] table [x=threads, y=bstm1k,  col sep=comma] {plots/aptos-batch-10k};
        \addplot [line width=0.25mm, color=red, mark=triangle] table [x=threads, y=bstm10k,  col sep=comma] {plots/aptos-batch-10k};
        %
        %
        \addplot [line width=0.25mm, black, mark=square] [error bars/.cd, y explicit,y dir=both,] table [x=threads, y=seq,  col sep=comma] {plots/aptos-batch-10k};
    \end{axis}
    \end{tikzpicture}
\end{minipage}
\caption{Comparison of BSTM and Sequential execution for block size $10^3$ (left) and $10^4$ (right). Aptos p2p transactions.}
\label{fig:all-aptos}
\end{figure*}
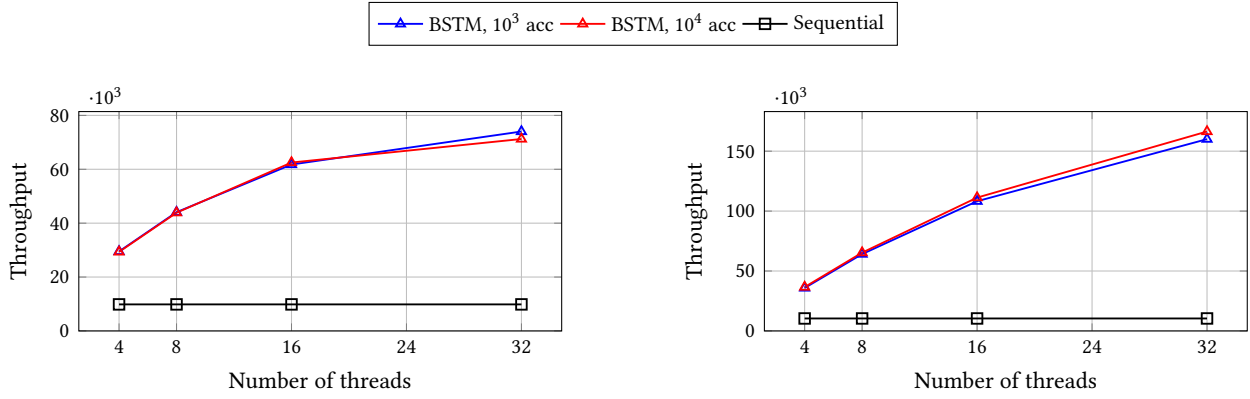

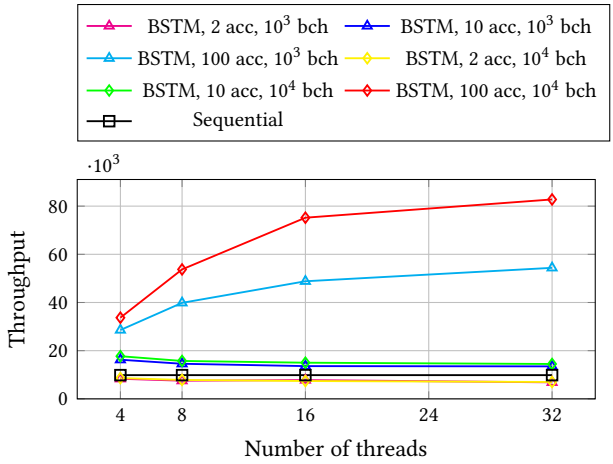
\begin{figure}[t]
    \centering
    \pgfplotsset{footnotesize,height=4.5cm, width=\linewidth}
    \begin{tikzpicture}
    \begin{axis}[
        scaled y ticks=base 10:-3,
        legend columns=2,
        legend style={font=\footnotesize, at={(0.0,1.8)},anchor=north west,legend columns=2},
        xlabel={Number of threads},
        ylabel= {Throughput},
        grid=major,
        xtick={4, 8, 16, 24, 32},
        ymin = 0,
        ]
        \addplot [line width=0.25mm, color=magenta, mark=triangle] table [x=threads, y=bstm2,  col sep=comma] {plots/aptos-batch-1k};
        \addplot [line width=0.25mm, color=blue, mark=triangle] table [x=threads, y=bstm10,  col sep=comma] {plots/aptos-batch-1k};
        \addplot [line width=0.25mm, color=cyan, mark=triangle] table [x=threads, y=bstm100,  col sep=comma] {plots/aptos-batch-1k};
        \addplot [line width=0.25mm, color=yellow, mark=diamond] table [x=threads, y=bstm2,  col sep=comma] {plots/aptos-batch-10k};
        \addplot [line width=0.25mm, color=green, mark=diamond] table [x=threads, y=bstm10,  col sep=comma] {plots/aptos-batch-10k};
        \addplot [line width=0.25mm, color=red, mark=diamond] table [x=threads, y=bstm100,  col sep=comma] {plots/aptos-batch-10k};
        \addplot [line width=0.25mm, black, mark=square] [error bars/.cd, y explicit,y dir=both,] table [x=threads, y=seq,  col sep=comma] {plots/aptos-batch-1k};
        \addlegendentry{BSTM, 2 acc, $10^3$ bch}
        \addlegendentry{BSTM, 10 acc, $10^3$ bch}
        \addlegendentry{BSTM, 100 acc, $10^3$ bch}
        \addlegendentry{BSTM, 2 acc, $10^4$ bch}
        \addlegendentry{BSTM, 10 acc, $10^4$ bch}
        \addlegendentry{BSTM, 100 acc, $10^4$ bch}
        \addlegendentry{Sequential}
    \end{axis}
    \end{tikzpicture}
    \caption{Comparison of BSTM and sequential execution for block size $10^3$ and $10^4$, account sizes 2, 10 and 100. Aptos p2p transactions.}
    \label{fig:bstm-high-aptos}
\end{figure}

\begin{figure}[t]
    \centering
    \pgfplotsset{footnotesize,height=4.5cm, width=\linewidth}
    \begin{tikzpicture}
    \begin{axis}[
        scaled ticks=base 10:-3,
        legend columns=2,
        legend style={font=\small, at={(0,1.5)},anchor=north west,legend columns=2},
        xlabel={Block size},
        ylabel= {Throughput},
        grid=major,
        xtick={1000, 5000, 10000, 20000, 50000},
        ]
        \addplot [line width=0.25mm, color=magenta, mark=triangle] table [x=batches, y=bstm-16-1k,  col sep=comma] {plots/aptos-batches};
        \addplot [line width=0.25mm, color=blue, mark=triangle] table [x=batches, y=bstm-16-10k,  col sep=comma] {plots/aptos-batches};
        \addplot [line width=0.25mm, color=cyan, mark=*] table [x=batches, y=bstm-32-1k,  col sep=comma] {plots/aptos-batches};
        \addplot [line width=0.25mm, color=purple, mark=*] table [x=batches, y=bstm-32-10k,  col sep=comma] {plots/aptos-batches};
        \addlegendentry{16 threads, $10^3$ acc}
        \addlegendentry{16 threads, $10^4$ acc}
        \addlegendentry{32 threads, $10^3$ acc}
        \addlegendentry{32 threads, $10^4$ acc}
    \end{axis}
    \end{tikzpicture}
    \caption{Throughput of BSTM for various block sizes. Aptos p2p transactions.}
    \label{fig:tps-aptos}
\end{figure}
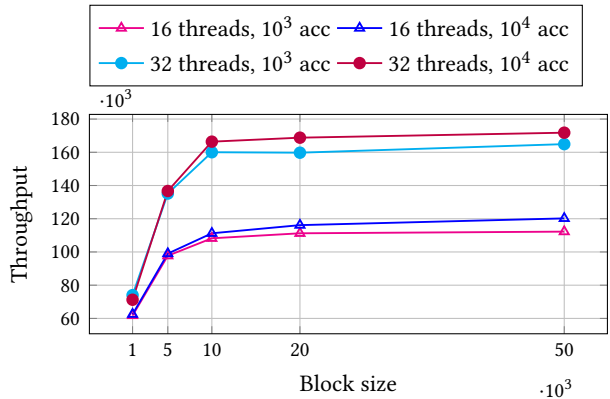
\section{Implementation and Evaluation}
\label{sec:ev}

Our \BSTM implementation is in Rust, and is merged on the main branch of the open source Diem and Aptos projects~\cite{diem, aptoswhitepaper}. Both Blockchains run a virtual machine for smart contracts in Move language~\cite{blackshear2019move}. The \module{VM} captures all execution errors that could stem from inconsistent reads during speculative transaction execution. The \module{VM} also caches the reads from \module{Storage}. Importantly, the preset order allows us to test correctness by comparing to sequential implementation outputs.

Diem VM does not support suspending transaction execution at the exact point when a read dependency is encountered. Instead, when a transaction is aborted due to a \rust{READ\_ERROR}, it is later (after the dependency is resolved) restarted from scratch. Aptos VM supports this feature. 


To mitigate the impact of restarting \module{VM} execution from scratch, we check the read-set of the previous incarnation for dependencies
before the \module{VM}.\rust{execute} invocation in~\lineref{line:vmexeccall}.

Another related optimization implemented in \BSTM occurs when the \module{Scheduler}.\rust{add\_dependency} invocation returns \emph{false} in~\lineref{line:suspendexec}.
This indicates that the dependency has been resolved. 
Instead of~\lineref{line:retry} (that would restart the execution from scratch with the Diem VM), \BSTM calls \rust{add\_dependency} from the \module{VM} itself, and can thus re-read and continue execution when \emph{false} is returned.

\BSTM implementation uses the standard cache padding technique to mitigate false sharing. The logic for \emph{num\_active\_tasks} is implemented using the Resource Acquisition Is Initialization (\href{https://doc.rust-lang.org/rust-by-example/scope/raii.html}{RAII}) design pattern. Finally, \BSTM implements the \emph{data} map in \module{MVMemory} as a concurrent hashmap over access paths, with lock-protected search trees for efficient \emph{txn\_idx}-based look-ups.

\subsection{Experimental Results}
We evaluated \BSTM on a Amazon Web Services c5a.16xlarge instance (AMD EPYC CPU and 128GB memory) with Ubuntu 18.04 operating system. The experiments run on a single socket with up to 32 physical cores without hyper-threading.

The evaluation benchmark executes the whole block, consisting of peer-to-peer (p2p) transactions implemented in Move. 
Each p2p transaction randomly chooses two different accounts and performs a payment.

We first perform experiments with \emph{Diem} p2p transactions
~\footnote{\url{https://github.com/danielxiangzl/Block-STM}} 
that perform $21$ reads and $4$ writes.
For a Diem p2p transaction from account $A$ to account $B$, the $4$ writes of the transaction involve updating balances and sequence numbers of $A$ and $B$.
The reason for $21$ reads is that every Diem transaction is verified against some on-chain information to decide whether the transaction should be processed, some of which is specific to p2p transactions. During this process, information such as the correct block time and whether or not the account is frozen is read.

We also perform experiments with \emph{Aptos} p2p transactions
~\footnote{\url{https://github.com/danielxiangzl/Block-STM/tree/aptos}} 
that perform $8$ reads and $5$ writes each, where the Aptos p2p transactions reduce many of the verification and on-chain reads mentioned above. 
The VM execution overhead of a single Diem p2p compared to a single Aptos p2p is about $100\%$, as will be shown in \figureref{fig:all-diem} and \figureref{fig:all-aptos}, the throughput of sequentially executing Diem and Aptos p2p transaction is about $5k$ and $10k$, respectively.
We experiment with block sizes of $10^3$ and $10^4$ transactions and the number of accounts of $2, 10, 100, 10^3$ and $10^4$.
The number of accounts determines the amount of conflicts, and in particular, with just $2$ accounts the load is inherently sequential (each transaction depends on the previous one).
Each data point is an average of 10 measurements.

This reported measurements include the cost of reading all required values from storage, and computing the outputs (i.e. all affected paths and the final values), but not persisting the outputs to \module{Storage}. The outputs are computed according to the \module{MVMemory}.\rust{snapshot} logic, but parallelized (per affected memory locations). 

We compare \BSTM to \textsc{Bohm}~\cite{faleiro2015rethinking} and LiTM~\cite{xia2019litm}. 
\textsc{Bohm} is a deterministic database engine that enforces a preset order by assuming transactions' write-sets are known.
\textsc{Bohm} has a pre-execution phase in which it uses the write-sets information to build a multi-version data-structure that captures the dependencies with respect to the preset order.
Then, \textsc{Bohm} executes transactions in parallel, delays any transaction that has unresolved read dependencies by buffering it in a concurrent queue, and resumes the execution once the dependencies are resolved.
Note that in the Blockchain use-case the assumption of knowing all write-sets in advance is not realistic, so to compare \BSTM to \textsc{Bohm} we artificially provide \textsc{Bohm} with perfect write-sets information. Note that our measurements of \textsc{Bohm} only include parallel execution but not the write-sets analysis, thus would be significantly better than the performance of \textsc{Bohm} in practice when the write-sets analysis time is non-negligible.
LiTM~\cite{xia2019litm}, a recent deterministic STM library, claims to outperform other deterministic STM approaches on the Problem Based Benchmark Suite~\cite{shun2012brief}. We describe LiTM in more detail in~\sectionref{sec:related}.
In order to have a uniform setting for comparison, we implemented both a variant of \textsc{Bohm}
\footnote{\url{https://github.com/danielxiangzl/Block-STM/tree/bohm}} 
and LiTM
\footnote{\url{https://github.com/danielxiangzl/Block-STM/tree/litm}} 
in Rust in the Diem Blockchain.

The \BSTM comparison to \textsc{Bohm}, LiTM and sequential baseline for Diem p2p transactions is shown in \figureref{fig:all-diem}.  
The \BSTM comparison to sequential baseline for Aptos p2p transactions is shown in \figureref{fig:all-aptos}. 
We will open source all our implementations and benchmarks to enable reproducible results. 

\textbf{Comparison to \textsc{Bohm}~\cite{faleiro2015rethinking}.}
The results show that \BSTM has comparable throughput to \textsc{Bohm} in most cases, and is significantly better with $32$ threads and $10^3$ block size.
Since \textsc{Bohm} relies on perfect write-sets information and thus perfect dependencies among all transactions, it can delay the execution of a transaction after all its dependencies have been executed, avoiding the overhead of aborting and re-execution. 
In contrast, \BSTM require no information about write dependencies prior to execution and therefore will incur aborts and re-execution.
Still, the performance of \BSTM is comparable to \textsc{Bohm}, implying the abort rates of \BSTM is substantially small, thanks to the run-time write-sets estimation and the low-overhead collaborative scheduler.
We also found the overhead of constructing the multi-version data-structure of \textsc{Bohm} significant compared to \BSTM, without which \textsc{Bohm}'s throughput will be slightly better than \BSTM.

\textbf{Comparison to LiTM~\cite{xia2019litm}.}
With $10^4$ accounts, \BSTM has around $3$-$4$x speedup over LiTM regardless of the block size or transactions type (standard or simplified).
With $10^3$ accounts, the speedup is larger (up to $25$x) over LiTM, confirming that \BSTM is less sensitive to conflicts.

\textbf{Comparison to sequential execution.}
For Diem and Aptos benchmarks, \BSTM scales almost perfectly under low contention, achieving up to $90k$ tps and $160k$ tps, 
which is $18$x and $16$x over the sequential execution, respectively.

\textbf{Comparison under highly contended workload.}
\figureref{fig:bstm-high-diem} and \figureref{fig:bstm-high-aptos} reports \BSTM evaluation results with highly contended workloads. 
With a completely sequential workload ($2$ accounts) \BSTM has at most $30\%$ overhead vs the sequential execution in both Diem and Aptos benchmarks.
With $10$ accounts \BSTM already outperforms the sequential execution and with $100$ accounts \BSTM gets up to $8$x speedup in both benchmarks.
Note that with $100$ accounts \BSTM does not scale beyond $16$ threads, suggesting that $16$ threads already utilize the inherent parallelism in such a highly contended workload.

\textbf{Maximum throughput of \BSTM}
We also evaluate \BSTM with increasing block sizes (up to $50k$) to find the maximum throughput of \BSTM in \figureref{fig:tps-diem} and \figureref{fig:tps-aptos}.
For $32$ threads, \BSTM achieves up to $110k$ tps for Diem p2p ($21$x speedup over sequential) and $170k$ tps for Aptos p2p ($17$x speedup over sequential).
For $16$ threads, \BSTM achieves up to $67k$ tps for Diem p2p ($13$x speedup) and $120k$ tps for Aptos p2p ($12$x speedup).

\textbf{Conclusion.}
Our evaluation demonstrates that \BSTM is adaptive to workload contention and utilizes the inherent parallelism therein.
For Aptos benchmark, it achieves over $160k$ tps on workloads with low contention, over $80k$ on workloads with high contention, and at most $30\%$ overhead on workload that are completely sequential.

\section{Related Work}
\label{sec:related}

\paragraph{The STM approach.}
The problem of atomically executing transactions in parallel in shared memory has been extensively studied in the literature in the past few decades in the context of STM libraries (e.g., \cite{herlihy1993transactional, shavit1997software, dice2006transactional, dragojevic2011stm, felber2008dynamic, herlihy2008transactional, guerraoui2006stmbench7}).
These libraries instrument the concurrent memory accesses associated with different transactions, detect and deal with conflicts, and provide the final outcome equivalent to executing transactions sequentially in some serialization order.
In the STM libraries based on optimistic concurrency control 
~\cite{kung1981optimistic, dice2006transactional}, threads repeatedly speculatively execute and validate transactions.
A successful validation commits and determines the transaction position in the serialization order.

By default, STM libraries do not guarantee the same outcome when transactions are re-executed multiple times.
This is unsuitable for Blockchain systems, as validators need to agree on the outcome of block execution.
Deterministic STM libraries~\cite{nguyen2014deterministic, ravichandran2014destm, vale2016pot}
guarantee a unique final state.

Due to required conflict bookkeeping and aborts, general-purpose STM libraries often suffer from performance limitations compared to custom-tailed solutions and are rarely deployed in production~\cite{cascaval2008software}.
However, STM performance can be dramatically improved by restricting it to specific use-cases~\cite{spiegelman2016transactional, herman2016type, laborde2019wait, hassan2014optimistic, elizarov2019loft}.
For the Blockchain use-case, the granularity is a block of transactions. Thus, unlike the general setting, \BSTM do not need to handle a long-lived stream of transactions that arrive at arbitrary times and commit them one by one. 
Moreover, thanks to the VM, the Blockchain use-case does not require opacity~\cite{guerraoui2007opacity}.

\textbf{Preset and deterministic order.}
There is prior work on designing STM libraries constrained to the predefined serialization order~\cite{mehrara2009parallelizing, von2007implicit, saad2019processing}. 
In~\cite{mehrara2009parallelizing, von2007implicit} each transaction is committed by a designated thread and thus the predefined order reduces resource utilization.
This is because threads have to stall until all previous transactions in the order are committed before they can commit their own.
Transactions in~\cite{saad2019processing} are also committed by designated threads, but they limit the stalling periods to only the latency of the commit via a complex forwarding locking mechanism and flat combining~\cite{hendler2010flat} based validation.

Deterministic STM libraries~\cite{nguyen2014deterministic, ravichandran2014destm, vale2016pot, xia2019litm} consider a less restricted case in which every execution of the same set of transaction produces the same final state.
The idea in the state-of-the-art~\cite{xia2019litm} is simple. All transactions are executed from the initial state and the maximum independent set of transaction (i.e., with no conflicts among them) is committed, arriving to a new state.
The remaining transaction are executed from the new state, the maximum independent set is committed, and so on. 
This approach thrives for low conflict workloads, but otherwise suffers from high overhead.

To summarize, in the context of STM literature, the (deterministic or preset) ordering constraints have been viewed as a ``curse", i.e. an extra requirement that the system needs to satisfy at the cost of added overhead. For the \BSTM approach, on the other hand, the preset order is the ``blessing" that the whole algorithm is centered around. 
In fact, the closest works to \BSTM in terms of how the preset serialization order is used to deal with conflicts are from  the databases literature.
Calvin~\cite{calvin} and \textsc{Bohm}~\cite{faleiro2015rethinking} use batches (akin to blocks) of transactions and their preset order to execute transactions when their read dependencies are resolved.
This is possible because, in the databases context, the write-sets of transactions are assumed to be known in advance.
This assumption is not suitable for Blockchains as smart contracts might encode an arbitrary logic.
Therefore, \BSTM does not require the write-set to be known and learns dependencies on the fly.

\textbf{Multi-version data-structures.}
Multi-version data structures are designed to avoid write conflicts~\cite{bernstein1983multiversion}.
They have a history of applications in the STM context~\cite{cachopo2006versioned, perelman2010maintaining}, some of which utilize optimistic concurrency control~\cite{bortnikov2017omid}.
The multi-version data-structure maps between memory locations and values that are indexed based on versions that are assigned to transactions via global version clock~\cite{riegel2006lazy, dice2006transactional, bortnikov2017omid}.

\textbf{Blockchain execution.}
The connection between STM techniques and parallel smart contract execution was explored in the past~\cite{dickerson2020adding, amiri2019parblockchain, anjana2021optsmart, anjana2020efficient}.
A \emph{miner-replay} paradigm was explored in~\cite{dickerson2020adding}, where miners parallelize block execution using a white-box STM library application that extracts the resulting serialization order as a ``fork-join’’ schedule. This schedule is sent alongside the new block proposal (via the consensus component) from miners to validators. 
After the block is proposed, validators utilize the fork-join schedule to deterministically replay the block. 
%
ParBlockchain~\cite{amiri2019parblockchain} introduced an \emph{order-execute} paradigm (OXII) for deterministic parallelism. The ordering stage is similar to the schedule preparation in~\cite{dickerson2020adding}, but the transaction dependency graph is computed without executing the block. OXII relies on read-write set being known in advance via static-analysis or on speculative pre-execution to generate the dependency graph among transactions.
OptSmart~\cite{anjana2021optsmart, anjana2020efficient} makes two improvements. First, the dependency graph is compressed to contain only transactions with dependencies; those that are not included may execute in parallel. Second, execution uses multi-versioned memory to mitigate write-write conflicts.

Hyperledger Fabric~\cite{androulaki2018hyperledger} and several related works~\cite{ruan2020transactional, sharma2019blurring} follow the execute-order-validate paradigm. As a result, the execution phase can abort unserializable transactions before ordering. Transactions in~\cite{chen2021forerunner} are pre-executed off the critical path to produce hints for final execution.

\vspace{-1mm}
\section{Summary}
\label{sec:disc}
This paper presents \BSTM, a parallel execution engine for the Blockchain use-case that achieves up to 170k tps with $32$ threads in our benchmarks.
For a fully sequential workload, it has a smaller than 30\% overhead, mitigating any potential performance attacks.
\BSTM relies on the write-sets of transactions' last incarnations 
to estimate dependencies and reduce wasted work.
If write-set pre-estimation was available, e.g., with a best effort static analysis, it could be similarly used by the first incarnation of a transaction. 
Moreover, using static analysis to find the best preset order is an interesting future direction.

\BSTM uses locking for synchronization in the \module{Scheduler} module. It is possible to use standard multicore techniques to avoid using locks, however, we did not see significant performance difference in our experiments.
Thus, we chose the design with locks for the ease of presentation.

In Blockchain systems, there is usually an associated ``gas'' cost to executing transactions. A single location for gas updates, could make any block inherently sequential. However, this issue is typically avoided by tracking gas natively, burning it or having specialized types or sharded implementation. 

As discussed in the \sectionref{sec:ev}, Diem VM currently does not support suspending and resuming transaction execution.
Once this feature is available, \BSTM can restart execution from the read that caused suspension upon encountering a dependency.
A potential optimization to go along with this feature is to  
validate the reads that happened during the execution prefix (before transaction was suspended) upon resumption.
This could allow earlier detection of impending aborts.

The current \BSTM implementation is not optimized for NUMA architectures or hyperthreading. Exploring these optimizations is another direction for future research.
Another interesting direction is to explore nesting techniques~\cite{moss1981nested} for transactional smart contract design.

\section*{Acknowledgment}
The authors would like to thank Sam Blackshear and Avery Ching for fruitful discussions.

\clearpage
\bibliography{bib}

\clearpage

\appendix


\section{Correctness}
\label{sec:proof}
We consider concurrent runs\footnote{Typically called executions in the literature, but we use the term \emph{run} to avoid a naming clash with transaction execution.} by threads, where each thread performs a sequence of atomic operations, and there is a global order in which these operations appear to take place. We use the term \emph{time} to refer to a point in this global order, i.e. a time $T$ determines for each thread the operations that it performed before~$T$.
\subsection{Life of a Version}
We say that validation of version $v = (j, i)$ starts anytime a validation task with version $v$ is returned to some thread $t$, either in~\lineref{line:tryexeccall} or in~\lineref{line:gettaskcall}. 
We say execution of version $v$
starts immediately after~\lineref{line:executing} is performed that sets the status of transaction $tx_j$ to \rust{EXECUTING(i)}. 
We say that the execution of version $v$ \emph{aborts} immediately 
after~\lineref{line:depaborting} is performed, and that
the validation of version $v$ \emph{aborts} immediately after~\lineref{line:valaborting} is performed.
In both cases, the transaction status is set to \rust{ABORTING(i)}.

After thread $t$ starts the execution of version $v$, an execution task with $v$ is returned either in~\lineref{line:reexeccall} or in~\lineref{line:gettaskcall}.
Thread $t$ then invokes the \rust{try\_execute} function for the execution task, which may invoke \rust{finish\_execution} procedure in~\lineref{line:finishexeccall}.
The \rust{finish\_execution} function is not called only when the execution aborts, in which case we say the execution \emph{finishes} at the same time when it aborts.
Similarly, after a validation starts, $t$ invokes \rust{needs\_reexecution} function for the validation task, which always invokes\\ \rust{finish\_validation} procedure in~\lineref{line:finvalcall}.

If~\lineref{line:extaskdec} (for execution) or~\lineref{line:valtaskdec} (for validation) is performed, then the corresponding validation or execution \emph{finishes} immediately before.
If these lines are not performed in \rust{finish\_execution} and in \rust{finish\_validation}, respectively, then the \rust{finish\_execution} invocation returns a validation task and the \rust{finish\_validation} invocation returns an execution task.
We say that such an execution \emph{finishes} immediately before the \rust{try\_execute} invocation returns in~\lineref{line:tryexeccall} (i.e. before validation starts for the version in the returned task).
Analogously, such a validation finishes immediately before a \rust{needs\_reexecution} invocation returns in~\lineref{line:reexeccall} (i.e. before execution starts for the version in the returned task).

An update to a transaction status is always performed by a thread while holding the corresponding lock. \figureref{fig:status} describes all possible status transitions.
For example, once \emph{txn\_status[j]} becomes \rust{EXECUTING(i)}, it can never be\\ \rust{READY\_TO\_EXECUTE(i)} at a later time. 
By the code, illustrated in the allowable transitions in~\figureref{fig:status}, we have

\begin{corollary}
\label{cor:order}
The following observations are true:
\begin{itemize}
\item The status of transaction $tx_j$ must be set to\\ \rust{READY\_TO\_EXECUTE(i)} in~\lineref{line:setready} before the execution of the version $v = (j,i)$ can start.

\item Any version $v= (j,i)$ can be executed at most once (by a thread that updates the status of transaction $tx_j$ to \rust{EXECUTING(i)} from \rust{READY\_TO\_EXECUTE(i)} to start the execution of $v$). 
Only the executing thread may update the status next, either to \rust{ABORTING(i)} in~\lineref{line:depaborting} or to \rust{EXECUTED(i)} in~\lineref{line:finex}.


\item The status of transaction $tx_j$ must be set to \rust{EXECUTED(i)} in~\lineref{line:finex} during the execution of version $v = (j,i)$ before any validation of $v$ can start. Once the status is set to \rust{EXECUTED(i)}, it can only be updated to \rust{ABORTING(i)} in~\lineref{line:valaborting} during a validation of $v$.

\item At most one execution or validation of version $v = (j,i)$ can abort, updating the status to \rust{ABORTING(i)} either in~\lineref{line:depaborting} from \rust{EXECUTING(i)} or in~\lineref{line:valaborting} from \rust{EXECUTED(i)}. The next update to the status of transaction $tx_j$ must be to \rust{READY\_TO\_EXECUTE(i+1)} in~\lineref{line:setready}.
\end{itemize}
\end{corollary}

\subsection{Safety}
We say that \emph{a pre-validation} of transaction $tx_j$ starts any time some thread $t$ performs a \emph{fetch\_and\_increment} operation, returning $j$, in~\lineref{line:fival}.
The pre-validation finishes immediately before $t$ performs~\lineref{line:taskvaldec}, if this line is performed.
Otherwise, by code, a validation task for transaction $tx_j$ is returned from the\\ \rust{next\_version\_to\_validate} function invocation.
In this case, pre-validation finishes immediately before the validation task is returned in~\lineref{line:gettaskcall}, i.e. before the corresponding validation starts.

\begin{definition}[Global Commit Index]
\label{def:commit}
The global commit index at time $T$ is defined as the minimum among all the following quantities at time $T$: 
\begin{itemize}
    \item \module{Scheduler}.validation\_idx
    \item all indices $j$, such that \module{Scheduler}.txn\_status[j].status $\neq$ \rust{EXECUTED}
    \item transaction indices with ongoing pre-validation
    \item transaction indices of versions with ongoing
    execution or validation
\end{itemize}
\end{definition}
We say that transactions $tx_0, \ldots, tx_k$ of the block are \emph{globally committed} at time $T$ if 
the global commit index at time $T$
is strictly greater than $k$. 
Next, we prove the essential properties of the commit definition.
\begin{claim}
\label{clm:commit}
If transaction $tx_k$ is committed at time $T$, then it is also committed at all times $T' \geq T$.
\end{claim}
\begin{proof}
We prove this claim by a simple inductive reasoning on time.
Specifically, for every time $T' \geq T$ we prove that $k$ is strictly less than the global commit index at time $T'$.
The base case for time $T$ follows from the Claim assumption.
For the inductive step, we suppose the assumption holds at time $T'$ and show that the~\definitionref{def:commit} still leads to a global commit index $>k$ when the next event after $T'$ takes effect.
\begin{itemize}
    \item The operation may change validation index from time $T'$ only in~\lineref{line:decvalidx}, which can be due to a call in~\lineref{line:decvalcall2} (during \rust{finish\_execution}) or in~\lineref{line:decvalcall1} (during \rust{finish\_validation}).
    In the first case, if \emph{validation\_idx} is reduced to value $j$, there must be an ongoing execution with transaction index $j$ at time $T'$.
    In the second case, there must be an ongoing validation with transaction index $j$ at time $T'$.
    Thus, in both cases, by inductive hypothesis, $j > k$. 

    \item The operation may change a status of transaction $tx_j$ from \rust{EXECUTED} only in~\lineref{line:valaborting}, in which case there is an ongoing validation with transaction index $j$ at time $T'$. Thus, by inductive hypothesis, $j > k$.
    
    \item A \emph{fetch-and-increment} operation in~\lineref{line:fival} may start a pre-validation of transaction $tx_j$. The \emph{validation\_idx} must have been $j$ at time $T'$ and by inductive hypothesis, $j > k$.
    
    \item If validation of a version $v$ with transaction index $j$ starts immediately after $T'$, then there must have been a pre-validation or an execution of version $v$ that ended immediately before, hence, that was ongoing at time $T'$. Thus, by inductive hypothesis, $j > k$.
    
    \item If an execution of a version $v$ with transaction index $j$ starts immediately after $T'$, then let us consider two cases:
    \begin{itemize}
        \item if an execution task was returned in \lineref{line:reexeccall}, then there was a validation of a version with index $j$ (previous incarnation) that ended immediately before, and hence, was ongoing at time $T'$. Thus, by inductive hypothesis, $j > k$.
        
        \item if an execution task was returned to some thread $t$ in \lineref{line:gettaskcall}, then, by the code, the status of transaction $tx_j$ must have been previously set to \rust{EXECUTING} by $t$. 
        By~\corollaryref{cor:order}, the status of transaction $tx_j$ may not change to \rust{EXECUTED} until $t$ starts the execution.
        Thus, since the status of transaction $tx_j$ is not \rust{EXECUTED} at time $T'$, by inductive hypothesis, $j > k$.
    \end{itemize}
\end{itemize}

Hence, the global commit index is monotonically non-decreasing with time. 
\end{proof}
Next, we prove some auxiliary claims regarding the interplay between transaction status and shared (execution and validation) indices.
\begin{claim}
\label{clm:validxstatus}
Suppose all transactions are eventually committed, and that at all times after $T$ the status of transaction $tx_j$ is \rust{EXECUTED}.
If no validation of a version of $tx_j$ starts after $T$, then the validation index must be $>j$ at all times after $T$.
\end{claim}
\begin{proof}
Let us assume for contradiction that \emph{validation\_idx} is at most $j$ at some time $T' \geq T$.
Since all transactions are eventually committed and due to~\claimref{clm:commit}, \emph{validation\_idx} must have value $\rust{BLOCK.size()} > j$ at some time after $T'$.
The validation index is only incremented in~\lineref{line:fival}, which is by definition a start of pre-validation.
Therefore, transaction $tx_j$ must start pre-validation after $T'$, and pre-validation must finish due to~\definitionref{clm:commit} since all transactions are eventually committed.
By the claim assumption, transaction $tx_j$'s status is \rust{EXECUTED}, so by code (due to~\lineref{line:executed}), pre-validation finish must lead to a start of a validation of a version of $tx_j$, giving the desired contradiction.
\end{proof}

\begin{claim}
\label{clm:mustvalidate}
Suppose all transactions are eventually committed, and $i$ is the highest incarnation of transaction $tx_j$ such that version $v = (j, i)$ is executed.
Then, $v$ must start validation after~\lineref{line:finex} is performed in the execution of $v$. 
\end{claim}
\begin{proof}
The execution of version $v$ sets the status of transaction $tx_j$ to \rust{EXECUTED(i)} in~\lineref{line:finex}.
The execution of $v$ eventually finishes due to~\definitionref{def:commit} and~\claimref{clm:commit}, as transaction $tx_j$ eventually commits.
If a validation task is returned in~\lineref{line:retvaltask}, then a validation of version $v$ starts immediately after execution finishes.
Otherwise, by \corollaryref{cor:order}, the status of transaction $tx_j$ will remain \rust{EXECUTED(i)} unless it is updated to \rust{ABORTING(i)} by some validation of $v$, which also concludes the proof of the claim.
If the status remains \rust{EXECUTED(i)} and a validation task is not returned, then validation index has a value at most $j$ after the status update in~\lineref{line:finex} due to \lineref{line:valcompfin} and \lineref{line:decvalcall2}.
Then,~\claimref{clm:validxstatus} implies that a validation must start after~\lineref{line:finex} is performed.
\end{proof}

Next, we establish the correctness invariant of the committed transactions.
When we refer to a \emph{sequential run} of all transactions, we mean the execution of transaction $tx_0$, followed by the execution of transaction $tx_1$, etc, for all transactions in the block.
\begin{lemma}
\label{lem:comsafe}
After all transactions are committed, \module{MVMemory} contains exactly the paths written in the sequential run of all transactions. Moreover, a \emph{read} of a path from \module{MVMemory} with txn\_idx $= \rust{BLOCK.size()}$ returns the same value as the contents of the path after the sequential run.
\end{lemma}
\begin{proof}
Suppose all transactions are eventually committed.
Since initial status for each transaction is \rust{READY\_TO\_EXECUTE}, while~\definitionref{clm:commit} requires status \rust{EXECUTED}, by the code, for each transaction $tx_j$ the version $(j,0)$ must start executing at some point. 
Also, due to the commit definition and~\claimref{clm:commit}, all executions that start must finish (in order for the transactions to eventually be committed).
In fact, by~\claimref{clm:commit} the total number of executions, validations and pre-validations must be finite and they must all finish.
For each transaction index $j$, let $m_j$ the the highest incarnation for which there is an execution of version $(j, m_j)$. 
By~\corollaryref{cor:order}, among the versions of transaction $tx_j$ that are executed, version $(j, m_j)$ is executed last.
We show by induction on $j$ that the execution of version $(j, m_j)$ reads the same paths and values from~\module{MVMemory} as the execution of transaction $tx_j$ would during the sequential run. Thus, at the end of version $(j, m_j)$ execution, all entries with transaction index $j$ in~\module{MVMemory} also correspond to the same paths and contain the same values as the write-set in the sequential run.

The base case holds because every read with $\emph{txn\_idx} = 0$ reads from storage.
Next, suppose the inductive claim holds for transactions $tx_0, \ldots, tx_k$. 
By~\claimref{clm:mustvalidate}, version $v_{k+1} = (k+1, m_{k+1})$ is validated at least once after~\lineref{line:finex} is performed during $v_{k+1}$'s (unique, by~\corollaryref{cor:order}) execution. 
Any validation of $v_{k+1}$ that starts also finishes in order for the global commit index to reach values above $k+1$.
Finally, no validation of version $v_{k+1}$ may abort, as this would set \emph{txn\_status[k+1]} to an \rust{ABORTING} status and prevent global commit index from reaching \rust{BLOCK.size()} without another incarnation of transaction $tx_{k+1}$, contradicting the maximality of $m_{k+1}$. 
Therefore, we only need to show that a value read at any access path during the validation of $v_{k+1}$ is the same as in the sequential run of transaction $tx_{k+1}$. Then, since the validation must succeed, the execution of $v_{k+1}$ must have read the same values, and produced a compatible output to the sequential run, proving the inductive step.

Let $\alpha$ be the validation of $v_{k+1}$ that starts last.
Let $p$ be any path read during $\alpha$, and let $v_p$ be the corresponding version observed during the \rust{validate\_read\_set} invocation that returned \emph{true} (if the read returned a \rust{READ\_ERROR} in~\lineref{line:mvalidest} then $\alpha$ would fail).
If $v_p = \bot$, then validation $\alpha$, and the corresponding execution of version $v_{k+1}$ both read from storage.
If $v_p$ is a version of some transaction $tx_j$, since \module{MVMemory} only reads values from lower transactions, we have $j < k+1$.
Version $v_p$ is written during a \rust{record} call invoked in~\lineref{line:recordcall} during an execution that sets the status of transaction $tx_j$ to an \rust{EXECUTED} status before finishing.
We show this must have been the last execution of $tx_j$ using a proof by contradiction.
Otherwise,
by~\corollaryref{cor:order}, a validation $\beta$ of the same version of $tx_j$ must follow and abort.
Thus, by code, before finishing, $\beta$ marks path $p$ as an \rust{ESTIMATE}, after it is read by $\alpha$. The \emph{validation\_idx} is then ensured to be at most $j$ in~\lineref{line:decvalcall1} in $\beta$, contradicting~\claimref{clm:validxstatus} (Due to~\claimref{clm:mustvalidate} the status of transaction $tx_{k+1}$ is set to \rust{EXECUTED(m_{k+1})} in~\lineref{line:finex} during the execution of $v_{k+1}$, before $\alpha$ starts. 
Since no validation of $v_{k+1}$ aborts, by~\corollaryref{cor:order}, \emph{txn\_status[k+1]} never changes from \rust{EXECUTED}).


Hence, if $v_p$ is a version of $tx_j$, then the value read from $p$ is in fact the value written at path $p$ during the execution of the last $tx_j$'s version $(j, m_j)$. By the induction hypothesis, this is the same value that transaction $tx_j$ writes at $p$ in the fully sequential run. To finish the proof, suppose for contradiction that in the sequential run transaction $tx_{k+1}$ reads a value written by transaction $tx_{j'}$ with $j' > j$.
The validation $\alpha$ did not observe any entry from $j'$ at path $p$, not even an \rust{ESTIMATE}. However, by induction hypothesis, during the execution of version $(j', m_{j'})$ the same value as in the sequential run must be written to path $p$. Therefore, after a read by $\alpha$, there is an execution of a version of transaction $tx_{j'}$ that sets \emph{wrote\_new\_path} to \emph{true} due to $p$ and decreases validation index by calling~\lineref{line:decvalcall2}. This again contradicts our assumption about $\alpha$ and completes the proof, as the argument when $v_p = \bot$ instead of $v_p = (j, m_j)$ is analogous.
%
\end{proof}
\subsubsection{Number of Active Tasks}
\label{sec:numtasks}
What is left is to show is the safety of the \rust{check\_done} mechanism for determining when the transactions are committed. 
The key is to understand the role of the \emph{num\_active\_tasks} variable in the \module{Scheduler} module.
The \emph{num\_active\_tasks} is initialized to $0$ and incremented in~\lineref{line:taskexinc} and~\lineref{line:taskvalinc}.
The increment in~\lineref{line:taskvalinc} is accounting for the pre-validation that starts with a \emph{fetch-and-increment} in the following line (\lineref{line:fival}).
The 
\emph{num\_active\_tasks} is decremented in~\lineref{line:taskvaldec} if no validation task corresponding to the fetched index is created.
Otherwise, pre-validation leads to a the start of a validation, and \emph{num\_active\_tasks} is decremented immediately after the validation finishes, in~\lineref{line:valtaskdec} (unless an execution task is created for the caller).
The logic for execution tasks is analogous, with one difference that an execution can also finish in~\lineref{line:depaborting}, in which case \emph{num\_active\_tasks} is decremented shortly after, in~\lineref{line:deptaskdec}.
When \rust{finish\_execution} or \rust{finish\_validation} functions return a new task to the caller, \emph{num\_active\_tasks} is left unchanged (instead of incrementing and decrementing that cancel out). 
It follows that \emph{num\_active\_tasks} is always non-negative. The following auxiliary claims establish useful properties of when the value becomes $0$.

\begin{claim}
\label{clm:exidxstatus}
Suppose the status of transaction $tx_j$ was set to \rust{READY\_TO\_EXECUTE} at time $T$, and did not change until a later time $T'$. 
If execution index was at most $j$ at some time between $T$ and $T'$, then either num\_active\_tasks $>0$ or execution\_idx $\leq j$ at time $T'$. 
\end{claim}
\begin{proof}
Let as assume for contradiction that at time $T'$\\ \emph{num\_active\_tasks} is $0$ and \emph{execution\_idx} is strictly larger than $j$, but that at some time between $T$ and $T'$, the execution index was at most $j$.
Since execution index reaches a value larger than $j$ by time $T'$, a \emph{fetch-and-increment} operation must have been performed in~\lineref{line:fiex} between $T$ and $T'$, returning $j$.
The \emph{num\_active\_tasks} counter is incremented in the previous line, in~\lineref{line:taskexinc} (this is very similar to the increment to account for pre-validation, while here it is an analogous pre-execution stage). Since the status is of transaction $tx_j$ remains \rust{READY\_TO\_EXECUTE} until $T'$, the only way to reduce \emph{num\_active\_tasks} to $0$ at time $T'$ is to perform the corresponding decrement, which by code, would occur only after an execution of a version of transaction $tx_j$ (due to~\lineref{line:ready}).
However, before an execution finishes (and then \emph{num\_active\_tasks} is decremented), it must perform~\lineref{line:ready} and since the status of transaction $tx_j$ is \rust{READY\_TO\_EXECUTE}, it must update the status to \rust{EXECUTING} in~\lineref{line:executing}, giving the desired contradiction with assumption in the claim. 
\end{proof}

\begin{lemma}
\label{lem:doneindex}
Suppose execution\_idx $\geq \rust{BLOCK.size()}$, validation\_idx $\geq \rust{BLOCK.size()}$ and num\_active\_tasks is $0$ simultaneously at time $T$. Then, all transactions are committed at time $T$.
\end{lemma}
\begin{proof}
As \emph{num\_active\_tasks} is $0$ at time $T$, there may not be an ongoing pre-validation, validation or execution at time $T$.
This is because an increment corresponding of \emph{num\_active\_tasks} always occurs before the start, while the decrement always occurs after the finish of the corresponding pre-validation, validation or execution. 
Next, we will prove that for any transaction index $j$, $\module{Scheduler}.\emph{txn\_status[j].status} = \rust{EXECUTED}$ at time $T$.
Then, by~\definitionref{def:commit}, the global commit index is equal to the \emph{validation\_index}, which is at least \rust{BLOCK.size()}, meaning that all transactions are committed at time $T$.
 
In the following, we prove by contradiction that all transactions must have an \rust{EXECUTED} status at time $T$. Suppose $j$ is the smallest index of a transaction with a non-\rust{EXECUTED} status. Consider three cases:
\begin{itemize}
    \item $\module{Scheduler}.\emph{txn\_status[j].status} = \rust{READY\_TO\_EXECUTE}$. We consider the time when the \rust{READY\_TO\_EXECUTE} status was last set for transaction $tx_j$ in~\lineref{line:ready}. This is due to a call either in~\lineref{line:depfree} or in~\lineref{line:readycallval}.
    \begin{itemize}
        \item Call in~\lineref{line:depfree}: there is an ongoing execution, which must finish in order for \emph{num\_active\_tasks} to be $0$ at time $T$.
        Before finishing,\\ the \rust{decrease\_execution\_idx} invocation in~\lineref{line:depexdec} ensures that the execution index has a value at most $j$. 
        Thus, by \claimref{clm:exidxstatus}, the execution index is at most $j$ at time $T$. A contradiction.
        
        \item Call in~\lineref{line:readycallval}: there is an ongoing validation which must finish in order for \emph{num\_active\_tasks} to be $0$ at time $T$.
        Before finishing, \emph{execution\_idx} must be observed in~\lineref{line:exidxfincheck} to be strictly higher than $j$, or we would get a contradiction with~\claimref{clm:exidxstatus}.
        But then, \rust{try\_incarnate} must be called in~\lineref{line:finincarnate}, which by code, would observe \rust{READY\_TO\_EXECUTE} status and update it to \rust{EXECUTING}, contradicting the status at time $T$. 
    \end{itemize}
    
    \item $\module{Scheduler}.\emph{txn\_status[j].status} = \rust{EXECUTING}$. By~\corollaryref{cor:order} and the definition of execution, 
    there must be an ongoing execution at time $T$ (of a version of $tx_j$ by the thread that set the status), which we already showed is impossible.
    
    \item $\module{Scheduler}.\emph{txn\_status[j].status} = \rust{ABORTING}$. Let $T'$ be the time when the \rust{ABORTING} status was last set for transaction $tx_j$, which can be in~\lineref{line:depaborting} or in~\lineref{line:valaborting}.
    \begin{itemize}
        \item call in~\lineref{line:depaborting} in an \rust{add\_dependency} invocation: in this case, \emph{txn\_idx} must be $j$ and the thread must be holding a lock on the status of a \emph{blocking\_txn\_idx}, which we will call $j'$.
        Because \module{MVMemory} only reads entries, including an \rust{ESTIMATE}, from lower transactions, and reading an \rust{ESTIMATE} is required for calling the \rust{add\_dependency} function, we have $j' < j$. Since~\lineref{line:depaborting} was performed, due to the check in~\lineref{line:depexcheck}, the status of transaction $tx_{j'}$  cannot be \rust{EXECUTED}, but by the minimality of $tx_j$ it must be \rust{EXECUTED} at time $T$.
        Therefore, an execution of a version of $tx_{j'}$ must invoke~\lineref{line:finex} between times $T'$ and $T$.
        This execution must finish in order for \emph{num\_active\_tasks} to be $0$ at time $T$, meaning that \rust{resume\_dependencies} invocation in~\lineref{line:resumecall} must be completed before $T$.
        However, due to locks, $tx_j$ is now a dependency of $tx_{j'}$, and this \rust{resume\_dependencies} invocation must update the status of transaction $tx_j$ to\\ \rust{READY\_TO\_EXECUTE} due to the call in~\lineref{line:depfree}, contradicting the status at time $T$. 
        
        \item call in~\lineref{line:valaborting}: there is an ongoing validation which must finish in order for \emph{num\_active\_tasks} to be $0$ at time $T$. Before finishing, the status must be updated to \rust{READY\_TO\_EXECUTE} due to the call in~\lineref{line:readycallval}, contradicting the status at time $T$.
        \qedhere
    \end{itemize}
\end{itemize}
\end{proof}

\subsubsection{Safety Guarantees}


\begin{lemma}
\label{lem:deccnt}
Let time $T$ be right before the operation in~\lineref{line:decexidx} or operation in~\lineref{line:decvalidx} by thread $t$ takes effect. 
Suppose num\_active\_tasks is $0$ at some time $T' \geq T$. Then, thread $t$ must have incremented decrease\_cnt (in~\lineref{line:deccntex} or in~\lineref{line:deccntval}) between times $T$ and $T'$. 
\end{lemma}
\begin{proof}
Performing \lineref{line:decexidx} as a part of\\ \rust{decrease\_execution\_idx} reduces \emph{execution\_idx} to the minimum of \emph{execution\_idx} and \emph{target\_idx}, while performing~\lineref{line:decvalidx} as a part of \rust{decrease\_validation\_idx} is similar for the \emph{validation\_idx}. The \rust{decrease\_execution\_idx} procedure is invoked only in~\lineref{line:depexdec} as a part of an ongoing execution, and accounting for this execution, \emph{num\_active\_task} must be at least $1$ during the whole invocation. Hence, in order for \emph{num\_active\_tasks} to become $0$, it must be decremented after the execution completes. Thus, $t$ must first complete \rust{decrease\_execution\_idx}, which includes performing~\lineref{line:deccntex}.

The \rust{decrease\_validation\_idx} procedure is invoked either as a part of validation that aborts, or as a part of execution when \emph{wrote\_new\_path} is \emph{true} in \rust{finish\_execution}. In both cases, \emph{num\_active\_tasks} is at least $1$ accounting for the ongoing validation or execution, since both finish after \rust{decrease\_validation\_idx} invocation completes.
Hence, in order for \emph{num\_active\_tasks} to become $0$, by code, $t$ must decrement it after it finishes execution of validation.
However, before doing so, it must perform~\lineref{line:deccntval} and return from the \rust{decrease\_validation\_idx} invocation.
\end{proof}

\begin{theorem}
\label{thm:safe}
If a thread joins after invoking the \rust{run} procedure, then all transactions are necessarily committed at that time.
\end{theorem}
\begin{proof}
The threads return from the \rust{run} invocation when they observe a $\emph{done\_marker}=\emph{true}$ in~\lineref{line:donemark}.
The \emph{done\_marker} is set to \emph{true} in~\lineref{line:setdone} after observing that $\emph{validation\_idx} \geq \rust{BLOCK\_SIZE}$, $\emph{execution\_idx} \geq \rust{BLOCK\_SIZE}$ and \emph{num\_active\_tasks} is $0$. These checks are not performed atomically, but instead a double-collect mechanism is used on the \emph{decrease\_count} variable, which is a monotonically non-decreasing counter. In particular, \rust{check\_done} confirms that \emph{decrease\_count} did not change (increase) while \emph{execution\_idx}, \emph{validation\_index} and \emph{num\_active\_tasks} were read.

Since a thread joined, \emph{decrease\_count} did not increase while it first observed \emph{execution\_idx} to be at least \rust{BLOCK\_SIZE} at time $T_1$, then observed \emph{validation\_idx} to be at least \rust{BLOCK\_SIZE} at time $T_2 > T_1$, and finally observed \emph{num\_active\_tasks} to be $0$ at time $T_3 > T_2$. 
We show by contradiction that \emph{num\_active\_tasks} was $0$ and \emph{execution\_idx} and \emph{validation\_idx} were still at least \rust{BLOCK.size()} simultaneously at $T_3$.
Assume by contradiction that $T_3$ does not have this property.
Thus, \emph{execution\_idx} must be decreased between $T_1$ and $T_3$ or \emph{validation\_idx} must be decreased between $T_2$ and $T_3$.
In both cases, we can apply~\lemmaref{lem:deccnt}, implying that \emph{decrease\_count} must have been incremented between $T_1$ and $T_3$, giving the desired contradiction.

Therefore, \rust{check\_done} only succeeds if the number of active tasks is $0$ while the execution index and the validation index are both at least \rust{BLOCK.size()} at the same time. By~\lemmaref{lem:doneindex} and, all transactions must be committed at this time.
\end{proof}

The \module{MVMemory}.\rust{snapshot} function internally calls \rust{read} with $\emph{txn\_id}=\rust{BLOCK.size()}$ for all affected paths. By~\theoremref{thm:safe} all transactions are committed after a thread joins, so~\lemmaref{lem:comsafe} implies the following
\begin{corollary}
A call to \module{MVMemory}.\emph{snapshot()} after a thread joins returns the exact same values at exact same paths as would be persisted at the end of a sequential run of all transactions. 
\end{corollary}

\subsection{Liveness}
We prove liveness under the assumption that every thread keeps taking steps until it joins\footnote{A standard assumption used to prove deadlock-freedom and starvation-freedom of algorithms, which are equivalent in our, single-shot, setting.} and that the \module{VM}.\rust{execute} is wait-free.
We start by formally defining \emph{pre-execution} in an analogous fashion to pre-validation.
A pre-execution of a transaction $tx_j$ starts any time some thread $t$ performs a \emph{fetch\_and\_increment} operation, returning $j$, in~\lineref{line:fiex}.
The pre-execution finishes immediately before $t$ performs~\lineref{line:taskexdec}, if this line is performed.
Otherwise, by code, an execution task for transaction $tx_j$ is returned from the \\ \rust{next\_version\_to\_execute} function invocation.
In this case, pre-execution finishes immediately before the execution task is returned in~\lineref{line:gettaskcall}, i.e. before the corresponding execution starts.
\begin{lemma}
\label{lem:finiteint}
There are finitely many pre-executions, executions, pre-validations and validations.
\end{lemma}
\begin{proof}
We prove the lemma by induction on transaction index, with a trivial base case (no pre-execution, execution, pre-validation or validation occurs for transactions with indices $<0$). For the inductive step, show that for any transaction index $k$ there are finitely many associated pre-executions, pre-validations, executions or validations.
For the inductive hypothesis, we only assume that there are finitely many executions and validations for versions of transactions indexed $<k$. It implies that after some finite time $T$:
\begin{itemize}
    \item[(a)] the execution index is never updated to a value $\leq k$ in~\lineref{line:decexidx}. The \rust{decrease\_execution\_idx} procedure is only called in~\lineref{line:depexdec} as a part of an ongoing execution of some transaction $tx_j$ when execution index is reduced to the minimum index of other transactions that depend on $tx_j$, all of which must have index $>j$ (as only higher-indexed transactions could have read from \module{MVMemory} an \rust{ESTIMATE} written during $tx_j$'s execution and become a dependency).
    
    \item[(b)] the entries in \module{MVMemory} for transactions indexed lower than $k$ never change. This holds because\\ \module{MVMemory}.\rust{record} invocation in~\lineref{line:recordcall} that affects entries with transaction index $j$, is, as defined, a part of transaction $tx_j$'s execution.
\end{itemize}

Due to (a), only one pre-execution of transaction $tx_k$ may start after time $T$, so there are finitely many pre-executions for $tx_k$ in total.
Next, we show that there is at most one validation of a version of $tx_k$ that aborts after time $T$.
If such a version exists, let $(k,i)$ be the first version that aborts after $T$.
Due to~\corollaryref{cor:order}, version $(k,i)$ may not abort more than once, and after it aborts, an execution of version $(k, i+1)$ must complete before any validation of version $(k, i+1)$ (or higher) starts.
However, no validation of version $(k, i+1)$ may abort, since by (b), the entries associated with transaction indices strictly smaller than $k$ no longer change in the multi-version data-structure,
i.e. \module{MVMemory}\rust{.validate\_read\_set} for a version whose execution started after $T$ necessarily returns \emph{true} in~\lineref{line:allvalidcall}. 
Thus, after some finite time no execution of a version of transaction $tx_k$ may start, as this only happens either following a pre-execution or a validation that aborts.
Moreover, we can now show that similar to (a) for the execution index, after some finite time, the validation index can never be reduced to a value $\leq k$ in~\lineref{line:decvalidx}. This is because the \rust{decrease\_validation\_idx} procedure is either called in~\lineref{line:decvalcall2}, when the \emph{validation\_idx} may be reduced to $j$ as a part of a transaction $tx_j$'s ongoing execution, or it is called in~\lineref{line:decvalcall1}, when the validation index may be reduced to $j+1$ as a part of a transaction $tx_j$'s ongoing validation.

Therefore, there are finitely many pre-validations of transaction $tx_k$ and as a result, no validation of a version of $tx_k$ may start after some finite time.
This is because a validation starts either following a pre-validation, or an execution of a version of $tx_k$.
As there are finitely many threads, we obtain that there are finitely many total pre-validations and pre-executions of transaction $tx_k$, as well as executions and validations versions of $tx_k$.
\end{proof}
In \BSTM, locks are used to protect statuses and dependencies for transactions. 
We now prove starvation-freedom for these locks.
\begin{claim}
\label{clm:starve}
If threads keep taking steps before they join, then any thread that keeps trying to acquire a lock eventually succeeds.
\end{claim}
\begin{proof}
A lock on transaction dependencies is acquired in~\lineref{line:lockdep} or in~\lineref{line:deplockswap}, both of which, by definition, occur as a part of some version's execution.
There are more cases of when a lock on a transaction status may be acquired.
The operations in~\lineref{line:statlockincarn} and in~\lineref{line:statlocknextval} are a part of a pre-execution of pre-validation of some transaction, respectively.
The lock may be acquired in~\lineref{line:lockstatready} in order to set the \rust{READY\_TO\_EXECUTE} status as a part of an ongoing execution (call to \rust{set\_ready\_status} in~\lineref{line:depfree}) or validation (call in~\lineref{line:readycallval}). 
The operation in~\lineref{line:finex} sets the status to \rust{EXECUTED} as a part of an ongoing execution, and the operation in~\lineref{line:valaborting} sets the status to \rust{ABORTING} as a part of an ongoing validation (that aborts).
The remaining two instances in~\lineref{line:depexcheck} and
in~\lineref{line:depaborting} occur as a part of a version's execution when a dependency is encountered, while the thread is also holding a lock on dependencies. These are the only instances when a thread may simultaneously hold more than one lock,
and also only the two operations within any critical section that may involve waiting.  
Because the acquisition order in these cases is unique (first the lock for dependencies, then for status) and all threads keep taking steps, a deadlock is therefore impossible.

Moreover, as described above, all acquisitions happen as a part of an ongoing pre-execution, pre-validation, execution or validation.
By~\lemmaref{lem:finiteint}, there are finite number of these, implying that in our setting, deadlock-freedom is equivalent to starvation-freedom, i.e. as long as threads keep taking steps, any thread that tries to acquire a lock in \BSTM must eventually succeed.
\end{proof}
Combining the above claims, we show  
\begin{corollary}
\label{cor:noint}
Suppose all threads keep taking steps before they join and \module{VM}.\rust{execute} is wait-free. Then, after some finite time, there may not be any ongoing pre-execution, pre-validation, execution or validation.
\end{corollary}
\begin{proof}
By~\lemmaref{lem:finiteint}, there are finitely many pre-executions, pre-validations, executions and validations.
Since all threads keep taking steps, to complete the proof we need to show that they all finish within finitely many steps of the invoking thread.
This is true because \module{VM}.\rust{execute} is assumed to be wait-free, lock are acquired within finitely many steps
by~\claimref{clm:starve}, and by code there is no other potential waiting involved in pre-execution, pre-validation, execution or validation.
\end{proof}
\begin{theorem}
If threads keep taking steps before they join and \module{VM}.\rust{execute} is wait-free, then all threads eventually join.
\end{theorem}
\begin{proof}
For contradiction, suppose some thread never joins.
By the theorem assumption, the thread keeps taking steps and by~\claimref{clm:starve}, it acquires all required locks within finitely many steps.
Moreover, since the \module{VM}.\rust{execute} function is wait-free, by~\corollaryref{cor:noint}, after some finite time there can be no ongoing pre-execution, pre-validation, execution or validation.
By code, the thread in this case must keep repeatedly entering the loop in~\lineref{line:spawnloop} and invoking \rust{next\_task} in~\lineref{line:gettaskcall}, while both the execution index and the validation index are always $\geq \rust{BLOCK.size}$ - otherwise, a pre-execution or pre-validation would commence.

Since \rust{decrease\_execution\_idx} and \rust{decrease\_validation\_idx} procedures are only invoked as a part of an ongoing execution or validation, respectively, after some finite time, this counter remains unchanged. Finally, by the mechanism that counts the active tasks, described in~\sectionref{sec:numtasks}, \emph{num\_active\_tasks} counts ongoing pre-executions, pre-validations, executions and validations. By code and since all threads keep taking steps before they join, the counter is always decremented after these finish. Since by~\lemmaref{lem:finiteint}, all pre-executions, pre-validations, executions and validations eventually finish, after some finite time the \emph{num\_active\_tasks} counter must always be $0$.

The thread that repeatedly invokes \rust{next\_task} must also repeatedly call \rust{check\_done} procedure. However, by the above, after some finite time it must set the \rust{done\_marker} to \emph{true} in~\lineref{line:setdone}. However, the next time the thread reaches~\lineref{line:spawnloop}, it will not enter the loop and join, proving the theorem by contradiction. 
\end{proof}

\end{document}